\documentclass{llncs}
\usepackage[dvipdfmx]{graphicx}
\usepackage[curve,matrix,arrow]{xy}
\usepackage{proof}
\usepackage{latexsym}

\usepackage{bbm}

\usepackage[cmex10]{amsmath}
\usepackage{amssymb}
\usepackage{mathabx}
\usepackage{comment}

\newtheorem{Defi}[theorem]{Definition}

\newcommand{\QED}{\hspace*{\fill}\vrule height6pt width6pt depth0pt}

\newcommand{\Lin}{\mathbf{Lin}}

\newcommand{\Stable}{\longrightarrow_{st}}
\newcommand{\Linear}{\longrightarrow_{lin}}

\newcommand{\Dom}{\mathsf{dom}}

\newcommand{\Rep}[3]{#1 \stackrel{#2}{\longrightarrow} #3}

\newcommand{\Tot}{\mathsf{Tot}}
\newcommand{\STot}{\mathsf{STot}}

\newcommand{\coh}{\,\raisebox{-.2em}{\shortstack{$\frown$\\[-.3em]$\smile$}}\,}
\newcommand{\scoh}{\,\raisebox{0em}{\shortstack{$\frown$\\[-.3em]$$}}\,}
\newcommand{\icoh}{\,\raisebox{-.2em}{\shortstack{$\smile$\\[-.3em]$\frown$}}\,}

\newcommand{\BO}[1]{\langle #1 \rangle}
\newcommand{\Tr}{\mathsf{tr}}
\newcommand{\Den}{\mathsf{den}}

\newcommand{\LR}{\mathcal{LR}}

\newcommand{\Fin}{\mathsf{fin}}
\newcommand{\Max}{\mathsf{max}}
\newcommand{\with}{\,\&\,}

\def\llto{\mathbin{-\mkern-3mu\circ}}              

\newcommand{\X}{{\boldsymbol{X}}}
\newcommand{\Y}{{\boldsymbol{Y}}}
\newcommand{\Z}{{\boldsymbol{Z}}}

\newcommand{\U}{{\boldsymbol{U}}}
\newcommand{\B}{{\boldsymbol{B}}}

\newcommand{\RR}{{\boldsymbol{R}}}

\newcommand{\Nat}{\mathbb{N}}

\newcommand{\Dyad}{\mathbb{D}}
\newcommand{\Real}{\mathbb{R}}

\newcommand{\XX}{\mathbb{X}}
\newcommand{\YY}{\mathbb{Y}}
\newcommand{\ZZ}{\mathbb{Z}}

\newcommand{\Iff}{\quad \Longleftrightarrow\quad}

\begin{document}

\title{Coherence Spaces and Uniform Continuity}

\author{Kei Matsumoto}
\institute{RIMS, Kyoto University\\
\texttt{kmtmt@kurims.kyoto-u.ac.jp}}

\maketitle

\begin{abstract}
In this paper, we consider a model of classical linear logic based on 
coherence spaces endowed with a notion of totality. 
If we restrict ourselves to total objects, 
each coherence space can be regarded as a uniform space and 
each linear map as a uniformly continuous function. 
The linear exponential comonad then assigns to each uniform space $\X$ 
the finest uniform space $\,!\,\X$ compatible with $\X$. 
By a standard realizability construction, it is possible to consider 
a theory of representations in our model. 
Each (separable, metrizable) uniform space, 
such as the real line $\Real$, can then be represented by 
(a partial surjecive map from) a coherence space with totality. 
The following holds under certain mild conditions: 
a function between uniform spaces $\XX$ and $\YY$ 
is uniformly continuous if and only if it is realized by 
a total linear map between the coherence spaces representing $\XX$ and $\YY$. 
\end{abstract}

\section{Introduction}
Since the inception of Scott's domain theory in 1960's, 
\emph{topology} and \emph{continuity} have been playing a prominent role
in denotational understanding of logic and computation. 
On the other hand, \emph{uniformity} and \emph{uniform continuity}
have not yet been explored so much.
The purpose of this paper is to bring them into the setting of 
denotational semantics 
by relating them to another denotational model: 
\emph{coherence spaces} and \emph{linear maps}. Our principal idea is that 
linear maps should be uniformly continuous, not just in analysis,
but also in denotational semantics. The following situation, 
typical for computable real functions (in the sense of \cite{Ko91}), illustrates 
our idea.

\begin{example}
Imagine that each real number $x \in \Real$ is presented by a 
rational Cauchy sequence $(x_n)_{n\in\Nat}$ with $|x - x_n| \leq 2^{-n}$.
Let $f: \Real \to \Real$ be a computable function 
which is uniformly continuous.
Then there must be a function $\mu : \Nat \to \Nat$,
called a \emph{modulus of continuity}, such that 
an approximation of $f(x)$ with precision $2^{-m}$
can be computed from a single rational number $x_{\mu(m)}$, 
no matter where $x$ is located on the real line.
Thus one has to access the sequence $(x_n)$ (regarded as an oracle)
only once.

On the other hand, 
if  $f: \Real \to \Real$ is not uniformly continuous, 
it admits no uniform modulus of continuity. Hence one has to accsess 
$(x_n)$ at least twice to obtain an approximation of $f(x)$,
once for figuring out the location of $x$
and thus obtaining 
a local modulus of continuity $\mu$ around $x$, 
once for getting the approximate value $x_{\mu(m)}$. 
\end{example}

Thus there is a difference in \emph{query complexity} between
uniformly continuous and non-uniformly continuous functions.
This leads us to an inspiration that linear maps, whose query complexity
is 1, should be somehow related to uniformly continuous functions.
To materialize this inspiration, we work with 
coherence spaces with totality.


\emph{Coherence spaces}, introduced by Girard \cite{Gi87}, 
are domains which are simply presented as undirected reflexive graphs.
It was originally introduced as a denotational semantics for System F, 
and later led to the discovery of linear logic. 
One of the notable features of coherence spaces 
is that there are two kinds of morphisms coexisting: stable and linear maps. 

\emph{Totalities}, which originate in domain theory
(eg.\ \cite{Gi86,No90,Be93}), are often attached to coherence spaces
(eg.\ \cite{KN97}). Specifically,
a \emph{coherence space with totality} in our sense is 
a coherence space $\X$ equipped with a set $\mathcal{T}_\X$ of cliques called a \emph{totality}, 
so that for any $a\in \mathcal{T}_\X$ there exists $\mathfrak{c}\in\mathcal{T}_{\X^{\perp}}$ with $a\cap\mathfrak{c}\neq\emptyset$, and vice versa. 
Totalities are usually employed to restrict objects and morphisms to total ones,
while we use them to impose a uniform structure on $\X$:
when restricted to ``strict" ones (to be defined later), 
a totality
$\mathcal{T}_\X$ can be seen as a set of 
\emph{ideal points} of a uniform space $\XX$, 
while a co-totality $\mathcal{T}_{\X^\perp}$ as the \emph{uniform sub-basis} 
for $\XX$.
Moreover, this allows us to prove that every 
``total" linear map $F: \X \Linear \Y$ is uniformly continuous
(though not vice versa).

The category of coherence spaces with totality and total linear maps 
forms a model of classical linear logic. 
In this setting, 
the linear exponential comonad $!$ admits an interesting interpretation:
it assigns to each uniform space $\X$ 
the \emph{finest} uniform space $\,!\,\X$ compatible with $\X$. 

We then apply our framework to \emph{computable analysis},
where people study computability over
various continuous and analytic structures
(such as the real numbers, metric spaces and topological spaces).
An essential prerequisite for this is that 
each abstract space should be \emph{concretely} represented.
While traditional approaches employ Baire spaces 
\cite{KW85,We00,BHW08} or
Scott-Ershov domains \cite{Bl97,ES99,SHT08}, 
we here consider representations based on coherence spaces.

This program has been already launched by  \cite{MT16}, 
where we have suitably defined \emph{admissible} representations 
based on coherence spaces (by importing various results from
the type-two theory of effectivity). The principal result 
there is as follows. Let $\XX$ and $\YY$ be topological spaces
admissibly represented by (partial surjections from) 
coherence spaces $\X$ and $\Y$ (eg.\ 
the real line $\Real$ is admissibly represented by a
coherence space $\RR$ in Example \ref{ex-real}).
Then a function $f: \XX \to \YY$ is sequentially continuous if and only if 
$f$ is realized (i.e., tracked) by a stable map
$F: \X \Stable \Y$.

In passing, we have also observed in \cite{MT16} a curious phenomenon:
when restricted to $\Real$, 
a function $f: \Real \to \Real$ is \emph{uniformly} continuous if and only if 
$f$ is realized by a \emph{linear} map $F: \RR \Linear \RR$.
Thus linearity in coherence spaces corresponds 
to uniform continuity of real functions. 
While we did not have any rationale or generalization,
at that time, we now have a better understanding of uniform continuity
in terms of coherence spaces. As a result, we 
are able to systematically generalize
the above result to \emph{separable metrizable uniform spaces}. 

\paragraph{Plan of the paper.}
We quickly review uniform spaces in \S 2.1
and coherence spaces in \S 2.2.
We then introduce in \S 3.1 
the notion of coherence space with totality, 
total and strict cliques,
and study the categorical structure. 
In \S 3.2, we explore the uniformities induced by co-totalities. 
In \S 4, we give an application of our model to computable analysis. 
We conclude in \S 5 with some future work. 

\section{Preliminaries}
\subsection{Uniform Spaces}
We review some concepts regarding uniform spaces. See \cite{Is64,Wi70} for details. 

A \emph{cover} of a set $X$ is a family of subsets $\mathcal{U}\subseteq \mathcal{P}(X)$ 
such that $\bigcup \mathcal{U} = X$. 
Let $\mathcal{U}$ and $\mathcal{V}$ be covers of $X$. 
We say that $\mathcal{U}$ \emph{refines} $\mathcal{V}$, 
written $\mathcal{U}\preceq\mathcal{V}$, 
if for every $U\in\mathcal{U}$ there exists $V\in\mathcal{V}$ with $U\subseteq V$. 
We then have the \emph{meet} (greatest lower bound) of $\mathcal{U}$ and $\mathcal{V}$ defined as 
$
\{U\cap V : \mbox{$U\in\mathcal{U}$ and $V\in\mathcal{V}$}\}
$,
denoted by $\mathcal{U}\wedge \mathcal{V}$. 

When $\mathcal{U}$ is a cover and $A$ is a subset of the set $X$, 
the \emph{star} $\mathrm{st}(A; \mathcal{U})$ is defined as 
$\bigcup\{U\in\mathcal{U}: A\cap U\neq\emptyset\}$.
Given any cover $\mathcal{U}$ of $X$, 
its \emph{star closure} is defined as 
$\mathcal{U}^* := \{\mathrm{st}(U;\mathcal{U}): U\in\mathcal{U}\}$, 
which is also a cover of $X$ and is refined by $\mathcal{U}$. 
We say that a cover $\mathcal{U}$ \emph{star-refines} $\mathcal{V}$ 
if $\mathcal{U}^* \preceq \mathcal{V}$. 
\begin{Defi}\label{d-uniformity}
A family $\mu$ of covers of $X$ is called a \emph{Hausdorff uniformity} if 
it satisfies the following: 
\begin{itemize}
\item[{\bf (U1)}] If $\mathcal{U},\mathcal{V}\in\mu$, then $\mathcal{U}\wedge\mathcal{V} \in\mu$; 
\item[{\bf (U2)}] If $\mathcal{U}\in \mu$ and $\mathcal{U}\preceq \mathcal{V}$, then $\mathcal{V}\in\mu$; 
\item[{\bf (U3)}] For every $\mathcal{U}\in\mu$, there exists $\mathcal{V}\in\mu$ which star-refines $\mathcal{U}$; 
\item[{\bf (U4)}] Given any two distinct points $x,y\in X$, there exists $\mathcal{U}\in\mu$ such that 
no $U\in\mathcal{U}$ contains both $x$ and $y$ (the \emph{Hausdorff condition}). 
\end{itemize}
\end{Defi}

Throughout this paper we always assume the Hausdorff condition.
A \emph{(Hausdorff) uniform space} is a pair $\XX=(X,\mu_X)$, a set $X$ endowed with a (Hausdorff) uniformity. 
Given any cover $\mathcal{U} \in \mu_X$ and any points $x,y\in X$, 
we write $|x-y|<\mathcal{U}$ if $x,y\in U$ for some $U\in\mathcal{U}$. 
The condition {\bf (U4)} can be restated as follows: 
if $|x-y|<\mathcal{U}$ for every $\mathcal{U}\in\mu_X$ then $x=y$. 

Let $\XX=(X,\mu_X)$ and $\YY=(Y,\mu_Y)$ be uniform spaces. 
A \emph{uniformly continuous} function from $\XX$ to $\YY$ 
is a function $f:X\to Y$ satisfying that
for any $\mathcal{V}\in\mu_Y$ there exists $\mathcal{U}\in\mu_X$ with 
$|x-y|<\mathcal{U} \ \Longrightarrow \ |f(x)-f(y)|<\mathcal{V}$
for every $x,y\in X$. 
A function $f:\XX\to \YY$ is called \emph{uniform quotient} if 
it is surjective and 
for every function $g: \YY\to\ZZ$ to a uniform space $\ZZ$, 
$g$ is uniformly continuous iff $g\circ f:\XX\to \ZZ$ is. 

A \emph{(uniform) basis} of a uniformity $\mu$ is a subfamily $\beta\subseteq \mu$ 
such that for every $\mathcal{U}\in\mu$ there exists $\mathcal{V}\in\beta$ with $\mathcal{V} \preceq \mathcal{U}$. 
A \emph{(uniform) sub-basis} of a uniformity $\mu$ is a subfamily $\sigma\subseteq \mu$  
such that the finite meets of members of $\sigma$ form a basis: 
for every $\mathcal{U}\in\mu$ there exist finitely many $\mathcal{V}_1, \ldots ,\mathcal{V}_n \in\sigma$ 
with $\mathcal{V}_1 \wedge \cdots \wedge \mathcal{V}_n \preceq \mathcal{U}$. 
Notice that if a family of covers satisfies the conditions {\bf (U2)}-{\bf (U4)} (resp. {\bf (U3)}-{\bf (U4)}), 
it uniquely generates a uniformity as a basis (resp. sub-basis). 

For instance, every metric space is in fact a uniform space. 
A uniformity on a metric space $\XX$ is generated by 
a countable basis 
$\mathcal{U}_n :=\{\mathcal{B}(x;2^{-n}) : \mbox{$x\in \XX$}\}$ ($n=1,2,\ldots$), 
where $\mathcal{B}(x;2^{-n})$ is the open ball of center $x$ and radius $2^{-n}$. 

On the other hand, every uniform space $\XX=(X,\mu_X)$ 
can be equipped with a topological structure, called the \emph{uniform topology}. 
A set $O\subseteq X$ is open with respect to the uniform topology iff 
for every $p\in O$ there exists $\mathcal{U}\in\mu_X$ such that 
$\mathrm{st}(\{p\},\mathcal{U})\subseteq O$. 
We will denote by $\tau_{\mathsf{ut}}(\mu)$ the uniform topology induced by a uniformity $\mu$. 
Given any uniformity $\mu$ on $X$, one can choose a basis $\beta$ consisting of open covers. 

It is easy to see that uniform continuity implies topological continuity: 
if a function $f:(X,\mu_X) \to (Y,\mu_Y)$ is uniformly continuous, then 
it is continuous as a function $f:(X,\tau_{\mathsf{ut}}(\mu_X)) \to (Y,\tau_{\mathsf{ut}}(\mu_Y))$. 

We say that a uniformity $\mu$ on $X$ is \emph{compatible} with a topology $\tau$ if 
$\tau= \tau_{\mathsf{ut}}(\mu)$. 
A topological space $\XX=(X,\tau)$ is said to be \emph{uniformizable} if 
there exists a uniformity $\mu$ on $X$ compatible with $\tau$. 
It is known that 
a topological space is uniformizable if and only it is Tychonoff.
For a metrizable space, the induced uniformity defined above is indeed compatible with 
the metric topology. 
In general, a uniformity is induced by a metric if and only if 
it has a countable basis.

Every Tychonoff (i.e. uniformizable) space $\XX=(X,\tau)$ 
can be equipped with the \emph{finest} uniformity $\mu_\mathrm{fine}$ which contains all of the uniformities compatible with $\tau$. 
A \emph{fine} uniform space is a uniform space endowed with the finest uniformity (compatible with its uniform topology). 
For a Tychonoff space $\XX=(X,\tau_X)$ 
we denote by $\XX_{\mathrm{fine}}=(X,\mu_{\mathrm{fine}})$ the fine uniform space 
compatible with $\tau_X$. 

The finest uniformity can be characterized as follows. 
Let $\mathbf{Tych}$ be the category of 
Tychonoff spaces and continuous maps, 
and $\mathbf{Unif}$ be the category of 
uniform spaces and uniformly continuous maps. 
The \emph{fine} functor $F:\mathbf{Tych} \longrightarrow \mathbf{Unif}$, 
which assigns to each Tychonoff space $\XX$ 
the fine uniform space $\XX_\mathrm{fine}$, 
is left adjoint to 
the \emph{topologizing} functor $G:\mathbf{Unif} \longrightarrow \mathbf{Tych}$, 
which assigns to each uniform space $\YY$ the topological space $\YY_{\mathsf{ut}}$ endowed with 
the uniform topology:
\begin{equation}
\xymatrix{
*+++{\mathbf{Tych}} \ar@/^10pt/[rr]^*++{F} & *+{\bot} & *+++{\mathbf{Unif}\ .} \ar@/^10pt/[ll]^*+{G} \\
}
\label{eq-tych-unif}
\end{equation}
Thus, for every Tychonoff space $\XX$ and uniform space $\YY$, 
$$
\mbox{$f:\XX \to \YY_{\mathsf{ut}}$ is continuous}\quad \Longleftrightarrow \quad \mbox{$f:\XX_\mathsf{fine}\to \YY$ is uniformly continuous}\ .
$$

\subsection{Coherence Spaces}
We here recall some basics of coherence spaces. 
See \cite{Gi87,Me09} for further information. 
\begin{Defi}
A \emph{coherence space} $\X=(X,\coh)$ consists of a set $X$ 
of tokens and a reflexive symmetric relation $\coh$ on $X$, 
called \emph{coherence}. 
\end{Defi}
Throughout this paper, we assume that every token set $X$ is countable. 
This assumption is quite reasonable in practice, 
since we would like to think of tokens as computational objects 
(see \cite{As90} for the study on computability over coherence spaces). 

A \emph{clique} of $\X$ is a set of pairwise coherent tokens in $X$.
By abuse of notation, we denote by $\X$ the set of all cliques of the coherence space $\X$. 
We also use the notations $\X_\Fin$ and $\X_\Max$ for the sets of 
all finite cliques and maximal cliques, respectively. 

Given tokens $x,y\in X$, we write $x\scoh y$ (\emph{strict coherence}) 
if $x\coh y$ and $x\neq y$. 
Notice that coherence and strict coherence are mutually definable from each other.
The coherence relation $\coh$ on the token set $X$ is naturally extended to $\X$ as: 
$a\coh b \iff a\cup b\in\X$ ($a,b\in\X$). 
This is equivalent to say that any token in $a$ is coherent with any token in $b$. 

An \emph{anti-clique} of $\X$ is a set of pairwise \emph{incoherent} tokens in X, 
that is, a subset $\mathfrak{a}\subseteq X$ such that $\neg (x\scoh y)$ for every $x,y\in \mathfrak{a}$. 
We will use the symbol $\icoh$ for incoherence: $x\icoh y \iff \neg (x\scoh y)$. 
Alternatively, an anti-clique of $\X$ is 
a clique of the \emph{dual} coherence space $\X^{\perp}:=(X,\icoh)$. 

It is known that the set $\X$ of cliques ordered by inclusion $\subseteq$ 
is in fact a Scott domain, whose compact elements are finite cliques of $\X$. 
Thus the \emph{Scott topology} on $\X$ is generated by $\{\BO{a}:a\in\X_\Fin\}$ as a basis, where 
$\BO{a}$ is an upper set defined by $\BO{a}:=\big\{b\in\X : a\subseteq b\big\}$.
We will denote by $\tau_{\mathrm{Sco}}$ this topology on $\X$. 

Given a subset $\mathcal{A} \subseteq \X$, we also write $\tau_{\mathrm{Sco}}$
for the induced subspace topology on $\mathcal{A}$.
Note that $\X$ is a $T_0$-space, and is countably-based due to the assumption that 
the token set $X$ is countable. 
Moreover, $(\X_\Max,\tau_{\mathrm{Sco}})$ is Hausdorff. 

Coherence spaces have a sufficiently rich structure to \emph{represent} abstract spaces. 
Let us begin with a coherence space for the real line $\Real$: 
\begin{example}[coherence space for real numbers]
\label{ex-real}
Let $\mathbb{D}:= \mathbb{Z}\times \mathbb{N}$, 
where each pair $(m,n)\in\Dyad$ is identified with 
a \emph{dyadic rational number} $m/2^{n}$. 
We use the following notations for $x=(m,n)\in\Dyad$: $\Den(x):=n$; 
$\mathbb{D}_n := \{x\in\Dyad \mid \Den(x)=n\}$ for each $n\in\Nat$; and 
$[x]:= [(m-1)/2^{n}; (m+1)/2^{n}]$. 

Hence $n=\Den(x)$ denotes the exponent of the denominator 
of $x$, and $[x]$ denotes the compact interval of $\Real$ 
with center $x$ and width $2^{-(n-1)}$. 

Let $\mathbf{R}$ be a coherence space $(\Dyad, \coh)$ defined by 
$x\scoh y$ iff $\Den(x)\neq\Den(y)$ and $[x]\cap [y]\neq \emptyset$.
The latter condition immediately implies the inequality $|x-y| \le 2^{-\Den(x)} + 2^{-\Den(y)}$, 
hence each maximal clique $a\in\mathbf{R}_\Max$ corresponds to 
a \emph{rapidly-converging} Cauchy sequence $\{x_n:n\in\Nat\}$ such that  
$x_n\in\Dyad_n$ for each $n\in\Nat$ and $|x_n -x_m| \le 2^{-n}+2^{-m}$ for every $n,m\in\Nat$. 

We then have a mapping $\rho_{\RR}:  \mathbf{R}_\Max\to \Real$ defined by 
$\rho_{\RR}(a):= \lim_{n\to \infty} x_n $. 
\end{example}

\begin{Defi}[stable and linear maps]
Let $\X$ and $\Y$ be coherence spaces. 
A function $F:\X\to \Y$ is said to be \emph{stable}, written $F:\X\Stable \Y$, 
if it is Scott-continuous and $a\coh b \Longrightarrow F(a\cap b) = F(a)\cap F(b)$ for any cliques $a,b\in\X$. 

A function $F:\X\to \Y$ is said to be \emph{linear}, written $F:\X\Linear \Y$, 
if it satisfies that $a = \sum_{i\in I} a_i \Longrightarrow F(a) = \sum_{i\in I}F(a_i)$, 
for any clique $a\in \X$ and any family of cliques $\{a_i\}_{i\in I}\subseteq \X$. 
Here $\sum$ means the disjoint union of cliques. 
\end{Defi}
It is easy to see that linearity implies stability.

There are alternative definitions. 
Given a function $F:\X\longrightarrow \Y$, call $(a,y)\in\X_\Fin \times Y$ 
a \emph{minimal pair} of $F$ if $F(a)\ni y$ and there is no proper subset 
$a'\subsetneq a$ such that $F(a')\ni y$. 
Denote by $\Tr(F)$ the set of all minimal pairs, called the \emph{trace} of $F$. 
Now, $F$ is a stable map iff it is $\subseteq$-monotone and satisfies that: 
if $F(a)\ni y$, there is a \emph{unique} $a_0 \subseteq a$ such that $(a_0, y)\in \Tr(F)$. 

If $F$ is furthermore linear, preservation of disjoint unions ensures that 
the finite part $a_0$ must be a singleton. 
Thus $F$ is a linear map iff it is $\subseteq$-monotone and satisfies that: 
if $F(a)\ni y$, there is a \emph{unique} $x \in a$ such that $(\{x\}, y)\in \Tr(F)$. 
By abuse of notation, we simply write $\Tr(F)$ for the set 
$\{(x,y)|(\{x\},y)\in\Tr(F)\}$ if $F$ is supposed to be linear. 

Below are some typical constructions of coherence spaces. 
Let $\X_i =(X_i ,\coh_i)$ be a coherence space for $i=1,2$. 
We define: 
\begin{itemize}
\item $\mathbf{1}:= \bot = (\{\bullet\},\{(\bullet,\bullet)\})$. 
\item $\X_1 \otimes \X_2 := (X_1 \times X_2 ,\coh)$, where 
$(z,x)\coh (w,y)$ holds iff both $z\coh_1 w$ and $x\coh_2 y$. 
\item $\X_1 \llto \X_2 := (X_1 \times X_2 ,\coh)$, where 
$(z,x)\scoh (w,y)$ holds iff $z\coh_1 w$ implies $x\scoh_2 y$. 
\item $!\,\X_1  := ((\X_1)_{\Fin} ,\coh)$, where 
$a \coh b$ holds iff $a\coh_1 b$. 
\end{itemize}
We omit the definitions of additives ($\with$ and $\top$). 
It easily follows that $\X^\perp \simeq \X\llto \perp$. 

A notable feature of coherence spaces 
is that they have two closed structures:
the category $\mathbf{Stab}$ of coherence spaces and stable maps 
is \emph{cartesian closed}; 
while the category $\Lin$ of 
coherence spaces and linear maps 
equipped with $(\mathbf{1},\otimes,\llto, \bot)$ 
is \emph{*-autonomous}. 
Moreover, the co-Kleisli category of the linear exponential comonad $!$ 
on $\Lin$ is isomorphic to $\mathbf{Stab}$ in such a way that 
a stable map $F:\X \Stable \Y$ can be
identified with a linear map $G:\,!\,\X\Linear \Y$
so that $\Tr(F)=\Tr(G) \subseteq \X_\Fin \times Y$. 
This leads to a 
linear-non-linear adjunction:
\begin{equation}
\xymatrix{
*+++{\mathbf{Stab}} \ar@/^10pt/[rr]^*++{K} &
 *+{\bot} &
  *+++{\Lin\ .} \ar@/^10pt/[ll]^*+{L} \\
}
\label{eq-stab-lin}
\end{equation}

The purpose of this paper is to establish a connection between 
the two adjunctions (\ref{eq-tych-unif}) and (\ref{eq-stab-lin}), 
which will be done in \S 3.2.

We do not describe the categorical structures in detail, but let us just mention the following. 
Given any linear map $F:\X\Linear\Y$, we have $\Tr(F)\in \X\llto \Y$. 
Conversely, given any clique $\kappa\in \X\llto\Y$, 
the induced linear map $\widehat{\kappa}:\X\Linear\Y$ is defined by
$\widehat{\kappa}(a):= \{y\in Y : \mbox{$(x,y)\in\kappa$ for some $x \in a$}\}$.

\section{Uniform Structures on Coherence Spaces}
In this section, we introduce a notion of (co-)totality 
on coherence spaces and observe that 
co-totality induces a uniform structure on the set of total cliques. 
\subsection{Coherence Spaces with Totality}
Let $\X$ be a coherence space. 
For any clique $a\in\X$ and any anti-clique $\mathfrak{c}\in\X^\perp$, 
$a\cap \mathfrak{c}$ is either empty or a singleton. If the latter is 
the case, we write $a\perp \mathfrak{c}$.

For any subset $\mathcal{A}\subseteq \X$, we write $\mathcal{A}^\perp$ for the set
$\{\mathfrak{c}\in\X^\perp :\mbox{$\forall a\in \mathcal{A}$. $a\perp \mathfrak{c}$}\}$ of anti-cliques of $\X$. 
One can immediately observe the following: 
(i) $\mathcal{A}\subseteq \mathcal{A}^{\perp\perp}\subseteq \X$; 
(ii) $\mathcal{B}\subseteq \mathcal{A}$ implies $\mathcal{A}^\perp \subseteq \mathcal{B}^\perp$; and
(iii) $\mathcal{A}^\perp = \mathcal{A}^{\perp\perp\perp}$. 
As a consequence, $\mathcal{A} = \mathcal{A}^{\perp\perp}$ iff
$\mathcal{A} = \mathcal{B}^\perp$ for some $\mathcal{B} \subseteq \X^\perp$.

\begin{Defi}[coherence spaces with totality]
A \emph{coherence space with totality} is 
a coherence space $\X$ endowed with a set $\mathcal{T}_\X \subseteq \X$
such that $\mathcal{T}_\X =\mathcal{T}_\X^{\perp\perp}$, called a \emph{totality}. 
Cliques in $\mathcal{T}_{\X}$ are said to be \emph{total}.
\end{Defi}

It is clear that a totality $\mathcal{T}_\X$ is 
upward-closed with respect to $\subseteq$, and is closed under 
compatible intersections: $a, b \in
\mathcal{T}_\X$ with $a \coh b$ implies $a\cap b \in \mathcal{T}_\X$.
As a consequence, every total clique $a \in \mathcal{T}_\X$ is associated 
with a unique minimal total clique 
$a^\circ := \bigcap \{b \in  \mathcal{T}_\X : b \subseteq a\} \in 
\mathcal{T}_\X$. Such a total clique is called 
\emph{strict} (or \emph{material} in the sense of ludics).
We write $\mathcal{T}_{\X}^{\circ}$ for the set of strict total cliques of $\X$. 
We have
$$
\mathcal{T}_\X = \{b \in \X : a \subseteq b \mbox{ for some } a \in \mathcal{T}_\X^\circ\}
= (\mathcal{T}_\X^\circ)^{\bot\bot}.
$$
Thus defining a totality is essentially equivalent to defining a strict totality. 
Notice that $a\in \mathcal{T}_\X$ iff 
for every $\mathfrak{c}\in\mathcal{T}_{\X}^{\perp}$,  $a\perp \mathfrak{c}$. 
In particular, $a\in \mathcal{T}_\X^\circ$ iff
for every $\mathfrak{c}\in(\mathcal{T}_\X^\perp)^\circ$ 
there exists $x\in a$ such that $x\in \mathfrak{c}$ and
dually, for every $x\in a$ there exists 
$\mathfrak{c}\in (\mathcal{T}_{\X}^\perp)^\circ$ such that $x\in \mathfrak{c}$.



Our use of totality is inspired by Kristiansen and Normann \cite{KN97}, 
although they use a set of anti-cliques of $!\X$ as totality
and they do not consider strictness and bi-orthogonality.
Similar constructions are abundant in the literature, eg.,
\emph{totality spaces} by Loader \cite{Loa94} and \emph{finiteness spaces} by Ehrhard \cite{Ehr05}. 

\begin{example}\label{ex-totality}
Consider the coherence space $\RR=(\Dyad,\coh)$ for real numbers defined in Example \ref{ex-real}. 
Then
$\mathcal{T}_\RR:= \RR_\Max$ is a totality on $\RR$: 
it is easy to see that $\mathcal{T}_{\RR}^{\perp} = \{\mathcal{D}_n :n\in\Nat\}$, 
hence $\mathcal{T}_{\RR}^{\perp\perp} = \RR_\Max = \mathcal{T}_\RR$. 
Moreover, $\mathcal{T}_{\RR}^\circ = \RR_\Max$ 
since $a^\circ \subseteq a$ and $a^\circ \in\RR_\Max$ imply $a^\circ =a$. 
\end{example}
\begin{example}\label{ex-standard-totality}
The idea of Example \ref{ex-real} can be generalized to a more general class. 
Let $\XX=(X,\mu)$ be a uniform space with a countable basis $\beta=\{\mathcal{U}_n\}_{n\in\Nat}$ 
consisting of countable covers. 
A \emph{metrization theorem} states that such a uniform space must be separable metrizable 
(see \cite{Ke75} for instance). 

Let $\B_\XX =(B,\coh)$ be a coherence space defined as $B=\coprod_{n\in\Nat}\mathcal{U}_n$ and 
$(n,U)\scoh (m,V)$ iff $n\neq m$ and $U\cap V\neq\emptyset$, 
where $\coprod_{n\in\Nat}\mathcal{U}_n$ means the \emph{coproduct} 
$\{(n,U):\mbox{$n\in\Nat$, $U\in\mathcal{U}_n$}\}$. 
Each $a\in (\B_\XX)_\Max$ corresponds to a sequence of members of uniform covers: 
$a=\{U_n\}_{n\in \Nat}$ such that $U_n \in\mathcal{U}_n$ for each $n\in\Nat$ and 
$U_n \cap U_m \neq\emptyset$ for every $n,m\in \Nat$. 
By the Hausdorff property, it indicates \emph{at most} one point in $\XX$. 

The separable metrizable space $\XX$ is represented by a partial map 
$\delta_\XX :\subseteq \B_\XX \to \XX$ defined by 
$\delta_\XX (a) := p\quad \Longleftrightarrow \quad p \in \bigcap_{n\in\Nat} U_n$, 
for every $p\in \XX$ and $a=\{U_n:n\in\Nat\}\in(\B_\XX)_\Max$. 
Let us define a totality by $\mathcal{T}_{\B_\XX} :=\Dom(\delta_\XX)^{\perp\perp}$. 

Notice that we do not have 
$\Dom(\delta_\XX) = \Dom(\delta_\XX)^{\perp\perp}$ in general, even though 
$\{\mathcal{U}_n :n\in\Nat\}\subseteq \mathcal{T}_{\B_\XX}^\perp$, 
since $\Dom(\delta_\XX )^{\perp\perp}=(\B_\XX)_\Max$. 
To make $\Dom(\delta_\XX)$ itself a totality, 
we have to assume that $\XX$ is \emph{complete} 
(every Cauchy sequence must be converging). 
\end{example}

All constructions of coherence spaces are extended with totality 
in a rather canonical way.
Let $\X=(X,\coh_\X)$ and $\Y=(Y,\coh_\Y)$ be coherence spaces, 
and $\mathcal{T}_\X \subseteq \X$ and 
$\mathcal{T}_\Y \subseteq \Y$ be totalities of $\X$ and $\Y$, respectively. 
Define: 
\begin{itemize}
\item $\mathcal{T}_{\X^\perp}:=\mathcal{T}_{\X}^{\perp}$; 
$\mathcal{T}_{\mathbf{1}}:= \mathbf{1}_\Max$. 
\item $\mathcal{T}_{\X\otimes \Y}:= (\mathcal{T}_{\X} \otimes \mathcal{T}_{\Y})^{\perp\perp}$,
where $a \otimes b := \{(x,y) : x\in a, y \in b\}$ for 
$a \in \X$ and $b\in \Y$, and 
$\mathcal{T}_{\X} \otimes \mathcal{T}_{\Y}$ is pointwise defined.
\item 
$
\mathcal{T}_{\X\llto \Y}:= \{\kappa \in (\X_1 \llto \X_2):\mbox{$\widehat{\kappa}[\mathcal{T}_\X]\subseteq \mathcal{T}_\Y$}$.
\item $\mathcal{T}_{\,!\,\X}:= (\,!\, \mathcal{T}_{\X})^{\perp\perp}$, where
$\,!\, a := \{a_0 \in \X: a_0 \subseteq_{\Fin} a\}$ for $a\in \X$,
and $!\, \mathcal{T}_{\X}$ is pointwise defined.
\end{itemize}

The connectives $\otimes$ and $!$ admit
``internal completeness'' in the following sense.

\begin{proposition}\label{prop-total-cliques}
$(\mathcal{T}_\X \otimes \mathcal{T}_\Y)^{\perp\perp\circ}= \mathcal{T}_{\X}^{\circ}\otimes \mathcal{T}_{\Y}^{\circ}$ holds whenever 
totalities $\mathcal{T}_{\X}$,
$\mathcal{T}_{\Y}$, 
$\mathcal{T}_{\X}^\perp$ and
$\mathcal{T}_{\Y}^\perp$ are all nonempty. 
$(\,!\,\mathcal{T}_{\Z})^{\perp\perp\circ}= \,!\, (\mathcal{T}_{\Z}^{\circ})$
holds 
for an arbitrary totality $\mathcal{T}_{\Z}$. 
\end{proposition}
A  proof is given in Appendix. 

Let us now turn to the morphisms.
\begin{Defi}
A linear map $F:\X\Linear \Y$ is called \emph{total}
if $\Tr(F)\in \mathcal{T}_{\X\llto\Y}$, 
or equivalently if $F$ preserves totality: $F[\mathcal{T}_\X]\subseteq \mathcal{T}_\Y$. 
\end{Defi}
A stable map $F:\X\Stable \Y$ is \emph{total} 
if so is the corresponding linear map $G:\,!\,\X\Linear \Y$ given in \S 2.2.

Denote by $\mathbf{Lin}_\Tot$ 
the category of coherence spaces with totality and total linear maps. 
It turns out to be a model of \emph{classical linear logic (CLL)}: 
\begin{theorem}\label{thm-lin-stot}
The category $\Lin_\Tot$ is a model of classical linear logic 
(i.e., a $\ast$-autonomous category with finite (co)products and a linear exponential (co)monad). 
\end{theorem}
This is due to Theorem 5.14 in \cite{HS03}. 
In fact, our construction of $\Lin_\Tot$ is essentially following the idea of 
\emph{tight orthogonality category} $\mathbf{T}(\Lin)$ induced by 
the orthogonality relation $\perp$, which can be shown to be a 
\emph{symmetric stable orthogonality} in $\Lin$. 

The category $\mathbf{Stab}_\Tot$ 
of coherence spaces with totality and total stable maps, 
is trivially the co-Kleisli category of 
the linear exponential comonad $\,!\,$ and hence is cartesian closed. 
%

\subsection{Uniformities induced by co-Totality}\label{subsec-unif}
We shall next show that each coherence space with totality 
can be equipped with a uniform structure. 
Our claim can be summarized as follows. 
Given a coherence space $\X$ with totality $\mathcal{T}_\X$,
the set of strict total cliques $\mathcal{T}_\X^\circ$ is 
endowed with both a topology and a uniformity: 
\begin{center}
{\it the totality $\mathcal{T}_\X^\circ$ is a set of points endowed with 
a Hausdorff topology $\tau_{\mathrm{Sco}}$\ ,}
\end{center}
while 
\begin{center}
{\it the co-totality $(\mathcal{T}_\X^{\perp})^\circ$ is a uniform sub-basis.}
\end{center}
Moreover, the co-totality $(\mathcal{T}_{\,!\,\X}^{\perp})^\circ$ on $\,!\,\X$
is a uniform basis, which induces the \emph{finest} uniformity on $\mathcal{T}_\X^\circ$.

Recall that each finite clique $a\in \X_\Fin$ generates 
the upper set $\BO{a}:=\{b\in \X: b\supseteq a\}$ in such a way that 
incoherence corresponds to disjointness: 
$$
\neg (x\coh y) \ \Longleftrightarrow \ \BO{x}\cap\BO{y}=\emptyset; \qquad\quad
\neg (a\coh b) \ \Longleftrightarrow \ \BO{a}\cap\BO{b}=\emptyset
$$
for every $x,y \in X$ and $a,b\in \X_\Fin$, where $\BO{x}$ stands for $\BO{\{x\}}$ by abuse of notation. Let us write 
$\BO{x}^\circ:= \BO{x} \cap \mathcal{T}_{\X}^\circ$ and
$\BO{a}^\circ:= \BO{a} \cap \mathcal{T}_{\X}^\circ$.


We call each $\mathfrak{c} \in (\mathcal{T}_{\X}^\perp)^\circ$ a \emph{uni-cover} of $\mathcal{T}_{\X}^\circ$. It 
can be seen as a disjoint cover 
$\{ \BO{x}^\circ : x \in \mathfrak{c}\}$
of $\mathcal{T}_{\X}^\circ$, since $\mathfrak{c}$ being total
precisely means 
that every $a\in\mathcal{T}_{\X}^\circ$ is contained in 
$\BO{x}^\circ$ for some $x\in\mathfrak{c}$.
Thus $\mathcal{T}_{\X}^\circ = \sum_{x\in\mathfrak{c}} \BO{x}^\circ$.
Moreover,
$\mathfrak{c}$ being strict means that 
$\BO{x}^\circ$ is nonempty 
for every $x\in \mathfrak{c}$.
That is, restricting $\mathfrak{c}
\in \mathcal{T}_{\X}^\perp$ to $\mathfrak{c}^\circ
\in (\mathcal{T}_{\X}^\perp)^\circ$ amounts to 
removing all empty $\BO{x}^\circ$ from the disjoint cover 
$\{ \BO{x}^\circ : x \in \mathfrak{c}\}$.

On the other hand, each $\mathfrak{C} \in (\mathcal{T}_{!\X}^\perp)^\circ$ 
is called an \emph{unbounded-cover} of $\mathcal{T}_{\X}^\circ$.
It is also identified with a disjoint cover
$\{\BO{a}^\circ: a\in \mathfrak{C}\}$ of $\mathcal{T}_{\X}^\circ$,
consisting of nonempty upper sets, so that  
$\mathcal{T}_{\X}^\circ =\sum_{a\in \mathfrak{C}} \BO{a}^\circ$. 

To emphasize the uniformity aspect,
we will use the notations $\sigma_{\X}^{\mathrm{b}}:= (\mathcal{T}_{\X}^\perp)^\circ$ and 
$\beta_{\X}^{\mathrm{ub}}:= (\mathcal{T}_{\,!\,\X}^\perp)^\circ$. 
Each uni-cover can be considered as an unbounded-cover consisting of singletons: 
$\sigma_{\X}^{\mathrm{b}} \subseteq \beta_{\X}^{\mathrm{ub}}$  
by $\mathfrak{c}\in \sigma_{\X}^{\mathrm{b}} \mapsto \{\{x\}: x\in\mathfrak{c}\}\in \beta_{\X}^{\mathrm{ub}}$. 


The families $\sigma_{\X}^\mathrm{b}$ and 
$\beta_{\X}^\mathrm{ub}$ indeed generate uniformities on $\mathcal{T}_{\X}^\circ$: 

\begin{proposition}
$(\mathcal{T}_\X^\circ, \beta_{\X}^\mathrm{ub})$
satisfies axioms {\rm ({\bf U1})}, 
{\rm ({\bf U3})} and {\rm ({\bf U4})}, while 
$(\mathcal{T}_\X^\circ, \sigma_{\X}^\mathrm{b})$
satisfies
{\rm ({\bf U3})} and {\rm ({\bf U4})}
in Definition \ref{d-uniformity}.
\end{proposition}

\begin{proof}
{\rm ({\bf U1})} 
Given $\mathfrak{A},\mathfrak{B}\in \beta^{\mathrm{ub}}_\X$, 
let $\mathfrak{A}\wedge\mathfrak{B} :=
\{a\cup b\ :\ \mbox{$a\in\mathfrak{A}$, $b\in\mathfrak{B}$ and $a\coh b$}\}^\circ$. It is indeed the meet of  $\mathfrak{A}$ and $\mathfrak{B}$, and belongs to
$\beta^{\mathrm{ub}}_\X = (!\mathcal{T}_\X^{\circ\bot})^\circ$. 
In fact, given $!c \in !\mathcal{T}_{\X}^\circ$, there are 
$a \in \mathfrak{A}$ and $b \in \mathfrak{B}$ such that $a \in !c$ and 
$b \in !c$. Hence $a\cup b \in !c \cap (\mathfrak{A} \wedge \mathfrak{B})$.\\
{\rm ({\bf U3})} 
In general, we have 
$\mathrm{st}(U,{\mathfrak{C}})=\bigcup \{V \in\mathfrak{C}: U\cap V\neq\emptyset\}= U$ 
for any disjoint cover ${\mathfrak{C}}$ of $\mathcal{T}_\X^\circ$ and
$U\in \mathfrak{C}$. Hence 
each $\mathfrak{A}\in \beta^{\mathrm{ub}}_\X$, which is disjoint,
star-refines itself. \\
{\rm ({\bf U4})} Assume that $a,b\in\mathcal{T}_{\X}^\circ$ with $a\neq b$. 
Then there are $x \in a \backslash b$ and $\mathfrak{c} \in 
\sigma_{\X}^\mathrm{b} = 
(\mathcal{T}_{\X}^\perp)^\circ$ 
such that $x \in \mathfrak{c}$ by strictness of $a$.
As  $a \in \BO{x}^\circ$,
$b \not\in \BO{x}^\circ$ and $\mathfrak{c}$ is a disjoint cover,
this witnesses the Hausdorff property for $\sigma_{\X}^\mathrm{b}$
(so for 
$\beta_{\X}^\mathrm{ub}$ too).
\QED
\end{proof}

Consequently, 
$\beta_{\X}^\mathrm{ub}$, as basis, 
generates a uniformity $\mu_{\X}^\mathrm{ub}$,
called the \emph{unbounded uniformity}, while
$\sigma_{\X}^\mathrm{b}$, as sub-basis,
generates another uniformity 
$\mu_{\X}^\mathrm{b} \subseteq \mu_{\X}^\mathrm{ub}$,
called the \emph{bounded uniformity}.
The index $\X$ will be often dropped if it is obvious from the context. 

As one may have noticed, the uniformities satisfy 
axiom {\rm ({\bf U3})} for a rather trivial reason. Nevertheless,
viewing coherence spaces with totality as uniform spaces 
will be \emph{essential} to establish our main theorem (Theorem \ref{thm-real-unif}).

Unlike $\beta^{\mathrm{ub}}_\X$, the set
$\sigma^{\mathrm{b}}_\X$ is not closed under finite meets. 
To make it closed, we have to extend it to 
another set $\beta^{\mathrm{b}}_\X \subseteq \beta^{\mathrm{ub}}_\X$ which consists of 
all finite meets of uni-covers: 
$\mathfrak{c}_1 \wedge \cdots \wedge \mathfrak{c}_m
:= \{ \{x_1, \dots, x_m\} \in \X : x_i \in \mathfrak{c}_i\ (1\leq i\leq m)\}^\circ$.
Notice that 
$\mathfrak{c}_1 \wedge \cdots \wedge \mathfrak{c}_m$
consists of cliques of size at most $m$. That is why 
$\mu_{\X}^\mathrm{b}$ is called bounded.

Although $\mu_{\X}^\mathrm{b}$ and $\mu_{\X}^\mathrm{ub}$ are 
different as uniformities, they do induce the same uniform topology.

\begin{proposition}
The (un)bounded uniformity on $\mathcal{T}_{\X}^\circ$ is compatible with
the Scott topology restricted to $\mathcal{T}_{\X}^\circ$. That is, 
$\tau_{\mathsf{ut}}(\mu^{\mathrm{b}})=\tau_{\mathsf{ut}}(\mu^{\mathrm{ub}}) =\tau_{\mathrm{Sco}}$.
\end{proposition}
\begin{proof}
By definition a set $U \subseteq \mathcal{T}_\X^\circ$ is open with respect
to $\tau_{\mathsf{ut}}(\mu^{\mathrm{ub}})$ iff for every $a\in U$
there exists $\mathfrak{A} \in \beta_{\X}^\mathrm{ub}$ such that 
$\mathrm{st}(\{a\};\mathfrak{A}) \subseteq U$ (see \S 2.1).
Due to disjointness of $\mathfrak{A}$, 
however, $\mathrm{st}(\{a\};\mathfrak{A})$ just amounts to 
$\BO{a_0}^\circ$, where 
$a_0$ is the unique clique in $\mathfrak{A}$ such that $a \in \BO{a_0}^\circ$.
Moreover, any $a_0\in \X_\Fin$ with 
$\BO{a_0}^\circ \neq \emptyset$ is contained in some 
$\mathfrak{A} \in \beta_{\X}^\mathrm{ub}$
by Lemma \ref{lemma-upper-extension} in \S \ref{A-thm-ubd-fine}. All together, 
$U$ is open iff for every $a\in U$
there exists $a_0 \in \X_\Fin$ such that $a \in \BO{a_0}^\circ$ 
iff $U$ is open with respect to $\tau_{\mathrm{Sco}}$.

The same reasoning works for $\tau_{\mathsf{ut}}(\mu^{\mathrm{b}})$ too.\QED
\end{proof}

The unbounded uniformity $\mu^{\mathrm{ub}}$ is hence 
compatible with, and finer than the bounded uniformity $\mu^{\mathrm{b}}$. 
We can furthermore show that it is the finest uniformity on $\mathcal{T}_{\X}^\circ$. 
The omitted proofs are found in \S \ref{A-thm-ubd-fine}. 
\begin{theorem}\label{thm-ubd-fine}
$(\mathcal{T}_{\X}^\circ, \mu^{\mathrm{ub}}_\X)$ is a fine uniform space. 
\end{theorem}

Due to the internal completeness (Proposition \ref{prop-total-cliques}), we have a bijection $\mathcal{T}_{\X}^\circ \simeq \mathcal{T}_{\,!\,\X}^\circ$ 
defined by $a\in \mathcal{T}_{\X}^\circ \leftrightarrow \,!\,a \in \mathcal{T}_{\,!\,\X}^\circ$. Notice also that 
$\beta^\mathrm{ub}_{\X} = (\mathcal{T}_{!\X}^\perp)^\circ
= \sigma^\mathrm{b}_{!\X}$ and 
fine uniformity is preserved under uniform homeomorphisms.
These facts together allow us to prove:

\begin{corollary}\label{prop-unif-hom}
There is a uniform homeomorphism 
$(\mathcal{T}_{\X}^\circ, \mu_{\X}^{\mathrm{ub}})\simeq (\mathcal{T}_{\,!\,\X}^\circ,\mu_{\,!\,\X}^{\mathrm{b}})$. As a consequence,
$(\mathcal{T}_{\,!\,\X}^\circ, \mu^\mathrm{b}_{\,!\,\X})$ is a fine uniform space. 
\end{corollary}


We are now ready to establish uniform continuity of linear maps.

\begin{theorem}\label{theo-strongly-unif}
A total linear map $F:\X\Linear\Y$ is \emph{strongly uniformly continuous}:  
for any $\mathfrak{b}\in \sigma_{\Y}^{\mathrm{b}}$ 
there exists $\mathfrak{a}\in\sigma_{\X}^{\mathrm{b}}$ 
such that $|a-b|<\mathfrak{a}\ \Rightarrow\ |F(a)-F(b)|<\mathfrak{b}$ 
for every $a,b\in\mathcal{T}_{\X}^\circ$. 
As a consequence:
\begin{itemize}
\item[(i)] Every total linear map $F:\X\Linear \Y$ is uniformly continuous 
w.r.t. the bounded uniformities.
\item[(ii)] Every total stable map $F:\X\Stable\Y$ is topologically continuous 
w.r.t. the uniform topologies. 
\end{itemize}
\end{theorem}
\begin{proof}
Note that the transpose $F^\bot : \Y^\bot \Linear \X^\bot$, 
defined by $x \in F^\bot(\{y\}) \Leftrightarrow F(\{x\})\ni y$
for every $x\in X$ and $y\in Y$,
is also total linear since 
$\Lin_\Tot$ is *-autonomous.
By linearity, any $x \in \mathfrak{a}$ is uniquely associated with
$y \in \mathfrak{b}$ such that $x \in F^\bot(\{y\})$
(i.e., $F(\{x\})\ni y$). From this, one can immediately observe
that $a,b \in \BO{x}^\circ$ with $x \in \mathfrak{a}$ 
implies $F(a), F(b) \in \BO{y}^\circ$ with $y \in \mathfrak{b}$.\QED
\end{proof}

We thus obtain a functor 
$J:\Lin_\Tot \to \mathbf{Unif}$ which sends
a coherence space with totality $(\X, \mathcal{T}_\X)$ 
to the uniform space $(\mathcal{T}_\X^\circ, \mu^\mathrm{b})$ 
and a total linear map to the corresponding uniformly continuous map 
which is shown in the above theorem.
There is also a functor 
$I:\mathbf{Stab}_\Tot \to \mathbf{Tych}$ sending 
$(\X, \mathcal{T}_\X)$ to the Tychonoff space 
$(\mathcal{T}_\X^\circ, \tau_\mathrm{Sco})$ and a total stable map 
to the corresponding continuous map.
We now have the following diagram, in which the two squares commute
(up to natural isomorphisms):
\begin{equation}
\xymatrix{
*+++{\mathbf{Tych}} \ar@/^10pt/[rr]^*++{F} & *+{\bot} & 
*+++{\mathbf{Unif}} \ar@/^10pt/[ll]^*+{G} \\
&&\\
*+++{\mathbf{Stab}_\Tot} \ar@/^10pt/[rr]^*++{K} 
\ar[uu]^*++{I} & *+{\bot} & 
*+++{\mathbf{\Lin}_{\Tot}\ .} \ar@/^10pt/[ll]^*+{L} \ar[uu]_*++{J} \ , 
}
\label{fig-pseudo}
\end{equation}
In addition, the pair of functors $\BO{I,J}$ 
preserves an adjunction: it is a \emph{pseudo-map of adjunctions} in the sense of Jacobs \cite{Ja99} 
(see Appendix in \S \ref{A-pm}).

This combines
 (\ref{eq-tych-unif}) and (\ref{eq-stab-lin}), as we have planned.

\section{Coherent Representations}
In this section, 
we exhibit 
a representation model based on coherence spaces 
and show that 
there exist good representations based on which linear maps well express uniformly continuous functions. 

\subsection{Representations as a Realizability Model}
We represent abstract spaces, 
largely following the mainstreams of \emph{computable analysis}: 
Baire-space representations in type-two theory of effectivity (TTE) \cite{KW85,We00,BHW08}, 
and domain representations \cite{Bl97,ES99,SHT08}. 
In both theories, computations are tracked by continuous maps 
over their base spaces (the Baire space $\mathbb{B}=\Nat^\omega$ for TTE or Scott domains for domain representations). 
Similarly we assign ``coherent'' representations to topological spaces, and 
track computations by stable maps, 
just as in Examples \ref{ex-real} and \ref{ex-standard-totality}. 

Let us formally give a definition: 
\begin{Defi}
Let $S$ be an arbitrary set. 
A tuple $(\X,\rho,S)$ is called a \emph{representation} of $S$ 
if $\X$ is a coherence space and $\rho:\subseteq \X\to S$ 
is a partial surjective function. 
Below, we write $\Rep{\X}{\rho}{S}$, 
or simply $\rho$ for $(\X,\rho,S)$. 
\end{Defi}
Representations enable us to express abstract functions as stable maps: 
\begin{Defi}[stable realizability]
Let $\Rep{\X}{\rho_\X}{S}$ and $\Rep{\Y}{\rho_\Y}{T}$ be representations. 
A function $f: S\to T$ is \emph{stably realizable}
with respect to $\rho_\X$ and $\rho_\Y$ 
if it is \emph{tracked} by a stable map $F:\X\Stable\Y$. 
That is, $F$ makes the following diagram commute: 
\begin{equation}
\xymatrix{
*+{\X} \ar[r]^*{F} \ar[d]_*{\rho_\X} & *+{\Y} \ar[d]^*{\rho_\Y} \\
 *+{S} \ar[r]^*{f} & *+{T}
}
\label{fig-realize}
\end{equation}
\end{Defi}
We denote by $\mathbf{StabRep}$
the category of coherent representations and stably realizable functions. 

With the help of Longley's theory of \emph{applicative morphisms} \cite{Lon94}, 
one can compare $\mathbf{StabRep}$ with other models of representations. 
By simply mimicking Bauer's approach \cite{Ba00,Ba02}, 
we obtain an \emph{applicative retraction} 
between coherent representations and TTE-representations. 
As a consequence, we can \emph{embed} TTE into the theory of coherent representations:
\begin{theorem}\label{theorem-retraction}
Let $\mathbf{TTERep}$ be a category which embodies TTE: 
the category of TTE-representations and \emph{continuously} realizable functions. 
Then $\mathbf{TTERep}$ is equivalent to a full coreflexive subcategory of $\mathbf{StabRep}$. 
\end{theorem}
For details on the realizability theory, we refer to \cite{Lon94}. 
We also refer to the Ph.D thesis of Bauer \cite{Ba00}, 
in which the relationship between the theory of (TTE and domain) representations and 
realizability theory is deeply studied. 

In \cite{MT16}, 
we have defined a full subcategory $\mathbf{SpStabRep}$ 
of $\mathbf{StabRep}$ 
which is equivalent to $\mathbf{TTERep}$, 
and introduced a concept of \emph{admissibility} of representations in $\mathbf{SpStabRep}$. 
The main result of \cite{MT16} is as follows: 
\begin{theorem}[\cite{MT16}]
\label{t-seqcont}
Let $\XX$ and $\YY$ be topological spaces represented by
admissible representations $\Rep{\X}{\rho_\X}{\XX}$ and $\Rep{\Y}{\rho_\Y}{\YY}$ 
in $\mathbf{SpStabRep}$. 
A function $f:\XX\longrightarrow \YY$ is stably realizable 
if and only if it is \emph{sequentially continuous}, that is, 
it preserves the limit of any convergent sequence:
$x_n \rightarrow x \ \Rightarrow\ 
f(x_n) \rightarrow f(x)$. 
\end{theorem}
For instance, the coherent representation $\Rep{\RR}{\rho_\RR}{\Real}$ defined in Example \ref{ex-real} 
belongs to $\mathbf{SpStabRep}$ and is admissible. 
Consequently, a function $f:\Real\to\Real$ is stably realizable w.r.t. $\rho_\RR$ 
iff it is continuous. 
This equivalence can be generalized 
to any countably-based $T_0$-space (and more generally, any \emph{qcb-space} in the sense of \cite{Si03}) 
as shown in \cite{We00,Sc02}. 

Notice that given any topological space $\XX$, its admissible representations are 
``interchangeable": if $\Rep{\X_0}{\rho_0}{\XX}$ and $\Rep{\X_1}{\rho_1}{\XX}$ are adimissible, 
then the identity map $\mathsf{id} : \XX
\longrightarrow \XX$ is realized by stable maps $F: \X_0 \Stable \X_1$
and $G: \X_1 \Stable \X_0$ 
which reduce each representation to another one. 

\subsection{Linear Realizability for Separable Metrizable Spaces}
On the other hand, we have found in \cite{MT16} a linear variant of the above equivalence 
between stable realizability and continuity: 
a function $f:\Real \to\Real$ is \emph{linearly realizable} iff it is \emph{uniformly continuous}. 
We below try to generalize this correspondence to 
a class of separable metrizable spaces, based on standard representations defined in Example \ref{ex-standard-totality}. 
\begin{Defi}[linear realizability]
Let $\Rep{\X}{\rho_\X}{S}$ and $\Rep{\Y}{\rho_\Y}{T}$ be representations. 
A function $f: S\to T$ is \emph{linearly realizable}
with respect to $\rho_\X$ and $\rho_\Y$ 
if it is \emph{tracked} by a linear map $F:\X\Linear\Y$. 
That is, $F$ makes the diagram (\ref{fig-realize}) commute. 
\end{Defi}
We denote by $\mathbf{LinRep}$ the category of coherent representations and 
linearly realizable functions. 

Given suitable totalities, a linear map $F$ which tracks $f$ turns out to be
uniformly continuous. 
First recall that for any set $\mathcal{A}\subseteq \X$ of a coherence space $\X$,
the set $\mathcal{A}^{\perp\perp}$ is a totality on $\X$, hence is endowed with 
a bounded uniformity (observed in \S 3). 
Here is an extension lemma for the double negation totalities:
\begin{lemma}\label{lem-tot-ext}
Let $\mathcal{A}\subseteq \X$ and $\mathcal{B}\subseteq \Y$ 
be arbitrary (non-empty) sets of cliques. 
If $F:\X\Linear \Y$ satisfies $F[\mathcal{A}]\subseteq \mathcal{B}$ 
then $F$ is indeed total: $F[\mathcal{A}^{\perp\perp}]\subseteq\mathcal{B}^{\perp\perp}$. 
\end{lemma}
Given any coherent representation $\Rep{\X}{\delta_\X}{S}$, 
let us endow $\X$ with a totality $\mathcal{T}_\X := \Dom(\delta_\X)^{\perp\perp}$. 
From the above lemma, we obtain that $f:S\to T$ is linearly realizable if and only if 
it is tracked by a total linear map $F:\X\Linear \Y$. 
So one can say that a linearly realizable function is in fact 
a ``totally linearly realizable'' function. 
Recall that $\Dom(\rho_{\RR})^{\perp\perp}=\Dom(\rho_{\RR})^{\perp\perp\circ}= \RR_\Max$ and 
$\Dom(\delta_{\XX})^{\perp\perp}= \Dom(\delta_{\XX})^{\perp\perp\circ}=(\B_{\XX})_\Max$. 

\begin{theorem}
$\mathbf{LinRep}$ is a linear category (i.e., a symmetric monoidal closed category with a linear exponential comonad). 
\end{theorem}
\begin{proofsketch}
Recall that a linear combinatory algebra (LCA) \cite{AHS02}
is a linear variant of well-known partial combinatory algebras (PCA). 
It is shown in Theorem 2.1 of \cite{AL05} that 
the PER category $\mathbf{PER}(\mathbb{A})$ over an LCA $\mathbb{A}$ 
is a linear category. 

We can naturally define an LCA $\mathbb{C}oh$ such that 
$\mathbf{LinRep}\simeq \mathbf{PER}(\mathbb{C}oh)$. 
Indeed, coherence spaces have linear type structures and 
there also exists a \emph{universal type}, from which 
we obtain an untyped LCA $\mathbb{C}oh$ 
by a linear variant of the Lietz-Streicher theorem \cite{LS02}. 

Consequently, the category $\mathbf{LinRep}\simeq\mathbf{PER}(\mathbb{C}oh)$ 
is a linear category. \QED
\end{proofsketch}
From the categorical structure of $\mathbf{PER}(\mathbb{C}oh)$, 
one can naturally construct various coherent representations, which are explicitly given in \S \ref{A-classical-rep}. 
We leave to future work to relate the co-Kleisli category $\mathbf{LinRep}_{\,!\,}$ 
and $\mathbf{StabRep}$. 

Then one can see that a standard representation $\Rep{\B_\XX}{\delta_\XX}{\XX}$ of 
a separable metrizable space $\XX$ 
is topologically ``good" for linear realizability, like admissible representations for stable realizability. 
It is shown in \S \ref{A-uni-realize} that a standard representation of $\XX$ 
does not depend on the chocie of a uniform basis, up to linear isomorphisms. 
Moreover, we can show that: 
\begin{theorem}\label{thm-unif-realize}
Let $\XX$ and $\YY$ be separable metrizable spaces with 
standard representations $\Rep{\B_\XX}{\delta_\XX}{\XX}$ and $\Rep{\B_\YY}{\delta_\YY}{\YY}$. 
Then every uniformly continuous function $f:\XX\to\YY$ is linearly realizable. 
\end{theorem}
See \S \ref{A-uni-realize} for a proof. 

For the other direction, we need a kind of connectedness in addition 
so that uni-covers of the coherence space 
exactly generates the uniformity on the represented space. 
A uniform space $\XX=(X,\mu_X)$ is \emph{chain-connected} 
(or sometimes called \emph{uniformly connected}) 
if for any two points $p,q\in X$ and every uniform cover $\mathcal{U}\in\mu_X$, 
there exist finitely many $U_1 ,\ldots U_n \in\mathcal{U}$ such that 
$p\in U_1$, $U_i \cap U_{i+1} \neq\emptyset$ for every $i<n$, and $U_n \ni q$. 
%

\begin{theorem}\label{thm-real-unif}
Let $\XX$ and $\YY$ be separable metrizable spaces represented by the standard representations. 
Provided that $\XX$ is chain-connected, 
a function $f:\XX\to\YY$ is linearly realizable iff it is uniformly continuous. 
\end{theorem}
\begin{proof}
The ``if''-direction is due to Theorem \ref{thm-unif-realize}. 
We shall show the ``only-if'' direction. 
As noted above, if $f$ is linearly realizable, there exists a total linear map $F:\B_\XX \Linear \B_\YY$ 
which tracks $f$, hence $F$ is uniformly continuous w.r.t. the bounded uniformities by Theorem 
\ref{theo-strongly-unif}. 
Any standard representation $\delta_\YY$ is also uniformly continuous 
as a partial map $\delta_\YY: \subseteq \B_\YY\to \YY$, 
so is the composition $\delta_\YY \circ F: \Dom(\delta_\XX) \to \YY$. 
Since $\delta_\XX$ is a uniform quotient by Lemma \ref{lem-quot} in \S \ref{A-lem-quot}, 
uniform continuity of $f\circ \delta_\XX = \delta_\YY \circ F$ 
implies that of $f$. \QED
\end{proof}
This result substantially and systematically 
generalizes the already mentioned result in \cite{MT16}: 
a function $f:\Real\to\Real$ is linearly realizable w.r.t. $\rho_\RR$ iff it is uniformly continuous. 

\section{Related and Future Work}

\paragraph{Type theory.}
In this paper, we have proposed 
coherence spaces with totality as an extension of 
ordinary coherence spaces, following the idea of Kristiansen and Normann. 
Originally in the domain theory, \emph{domains with totality}, are introduced by Berger \cite{Be93}
to interpret Martin-L\"{o}f type theory (i.e., intuitionistic type theory), using ``total'' domain elements. 
Since our model of coherence spaces with totality 
is a linear version of this model, 
one can expect that it could model \emph{intuitionistic linear type theory}.

Our theory also includes a natural representation of 
(separable, metrizable) uniform spaces and uniformly continuous maps
between them. Hence it might lead to a denotational model
of real functional programming languages 
(e.g., \cite{Es96,ES14}) extended with other uniform spaces,
where one can deal with uniformly continuous functions based on 
linear types.

\paragraph{Realizability theory.}
In the traditional setting, giving representations roughly amounts to 
constructing modest sets over a partial combinatory algebra (PCA) 
in the theory of realizability. 
Our model of coherent representations and stable realizability 
is in fact considered as a modest set model over 
a PCA $\mathbb{C}oh$, constructed 
from the \emph{universal coherence space} $\U$ in 
$\mathbf{Stab}$, to which one can embed any coherence spaces by \emph{linear}
(hence stable) maps. 
A modest set model turns out to be a model of intuitionistic logic \cite{Lon94,Ba00}. 
Bauer then gave an attractive paradigm \cite{Ba05}: 
\begin{center}
{\it Computable mathematics is a realizability interpretation\\ of 
constructive mathematics.}
\end{center}

On the other hand, less is known about the relationship between 
computable mathematics and linear realizability theory over 
a \emph{linear combinatory algebra (LCA)} \cite{AL00}, 
which is a linear analogue of PCA, and 
for which we can build a PER model of \emph{intuitionitstic linear logic}. 
Since the above 
universal coherence space $\U$ in fact resides in $\Lin$, it is in principle
possible to develop such a theory based on our framework.
We believe that exploring this direction, already mentioned 
in \cite{Ba00}, will be an interesting avenue for future work. 

\section*{Acknowledgement}
The author is greatful to 
Naohiko Hoshino and Kazushige Terui (RIMS) for useful comments.

\appendix
\section{Miscellaneous Proofs}
\subsection {Construction of Totalities}
\begin{lemma}
The functional totality is well-defined: $\mathcal{T}_{\X\llto\Y} = \mathcal{T}_{\X\llto\Y}^{\perp\perp}$.
\end{lemma}
\begin{proof}
Notice that $(\X\llto\Y)^\perp = \X\otimes \Y^\perp$ 
and 
$\widehat{\kappa}(a) \perp \mathfrak{c}$ iff
$\kappa \perp a \otimes \mathfrak{c}$ for 
any $\kappa \in \X\llto\Y$, $a \in \X$ and $\mathfrak{c} \in \Y^\bot$.
It follows that $\kappa \in \mathcal{T}_{\X \llto \Y}$ iff 
$\widehat{\kappa} (a) \perp \mathfrak{c}$ for any 
$a \in \X$ and $\mathfrak{c} \in \mathcal{T}_\Y^\bot$ iff $\kappa \in 
(\mathcal{T}_{\X} \otimes \mathcal{T}^\perp_{\Y})^\perp$.\QED
\end{proof}

The following lemmas prove the internal completeness of $\otimes$ and $\,!\,$ 
(Proposition \ref{prop-total-cliques}). 

\begin{lemma}\label{l-1}
Given $\mathfrak{a} \in \mathcal{T}_\X^\perp$ and
$\mathfrak{b} \in \mathcal{T}_\Y^\perp$, let 
$\mathfrak{a} \bullet \mathfrak{b} := \{ (x,y) : x\in 
\mathfrak{a},\ y \in \mathfrak{b}\}$. 
Then  
$\mathfrak{a} \bullet \mathfrak{b} \in
(\mathcal{T}_{\X} \otimes \mathcal{T}_{\Y})^\perp$.
\end{lemma}

\begin{proof}
First of all, $\mathfrak{a} \bullet \mathfrak{b}$ is an anti-clique 
of $\X \otimes \Y$.
Indeed, given $(x,y), (x', y') \in
\mathfrak{a} \bullet \mathfrak{b}$ with $(x,y) \neq (x', y')$, 
either $x\neq x'$ or $y\neq y'$. Assume that $x\neq x'$.
We then have $\neg (x \coh x')$, so 
$\neg((x,y) \coh (x', y'))$.

Now given $c \in \mathcal{T}_\X$ and $d \in \mathcal{T}_\Y$, we have
$c \perp \mathfrak{a}$ and $d \perp \mathfrak{b}$, from which we 
conclude 
$c \otimes d \perp \mathfrak{a} \bullet \mathfrak{b}$.
\QED
\end{proof}

\begin{lemma}\label{l-11}
$\mathcal{T}_{\X}^{\circ}\otimes \mathcal{T}_{\Y}^{\circ} \subseteq
(\mathcal{T}_\X \otimes \mathcal{T}_\Y)^{\perp\perp\circ}$.
\end{lemma}

\begin{proof}
Let $a \in \mathcal{T}_{\X}^{\circ}$ and 
$b \in \mathcal{T}_{\Y}^{\circ}$. It is clear that 
$a\otimes b \in 
(\mathcal{T}_\X \otimes \mathcal{T}_\Y)^{\perp\perp}$. Too see strictness,
let $(x,y) \in a \otimes b$. Then there are 
$\mathfrak{c} \in \mathcal{T}_\X^\perp$ and 
$\mathfrak{d} \in \mathcal{T}_\Y^\perp$ such that 
$x \in \mathfrak{c}$ and $y \in \mathfrak{d}$.
Hence 
$(x,y) \in \mathfrak{c}\bullet\mathfrak{d}$ and 
$\mathfrak{c}\bullet\mathfrak{d} \in 
(\mathcal{T}_{\X} \otimes \mathcal{T}_{\Y})^\perp$ by Lemma \ref{l-1}.
\QED
\end{proof}

\begin{lemma}\label{l-12}
Assume that $\mathcal{T}_{\X}^\perp$ and 
$\mathcal{T}_{\Y}^\perp$ are not empty.
Given $c \in (\mathcal{T}_{\X} \otimes 
\mathcal{T}_{\Y})^{\perp\perp}$, let 
$c^1 := \{x: (x,y) \in c \mbox{ for some } y\}$
and 
$c^2 := \{y: (x,y) \in c \mbox{ for some } x\}$.
Then 
$c^1 \in \mathcal{T}_\X$ and 
$c^2 \in \mathcal{T}_\Y$.
\end{lemma}

\begin{proof}
Let $\mathfrak{a} \in \mathcal{T}_\X^\perp$.
Since we suppose that $\mathcal{T}_\Y\neq\emptyset$ 
there is $\mathfrak{b} \in \mathcal{T}_\Y^\perp$ 
and $\mathfrak{a}\bullet\mathfrak{b} \in 
(\mathcal{T}_\X \otimes \mathcal{T}_\Y)^\perp$ by Lemma \ref{l-1}.
Hence $c \perp \mathfrak{a}\bullet\mathfrak{b}$ and 
we conclude that $c^1 \perp \mathfrak{a}$.\QED
\end{proof}

\begin{lemma}\label{l-13}
Assume that $\mathcal{T}_{\X}$ and 
$\mathcal{T}_{\Y}$ are not empty.
Given $\mathfrak{c} \in 
(\mathcal{T}_\X \otimes \mathcal{T}_\Y)^\perp$, let 
$\mathfrak{c}^1 := \{x: (x,y) \in \mathfrak{c} \mbox{ for some } y\}$
and 
$\mathfrak{c}^2 := \{y: (x,y) \in \mathfrak{c} \mbox{ for some } x\}$.
Then 
$\mathfrak{c}^1 \in \mathcal{T}_\X^\perp$ and 
$\mathfrak{c}^2 \in \mathcal{T}_\Y^\perp$.
\end{lemma}

\begin{proof}
Similarly.\QED
\end{proof}

\begin{lemma}\label{l-14}
Assume that $\mathcal{T}_{\X}$,
$\mathcal{T}_{\Y}$, 
$\mathcal{T}_{\X}^\perp$ and
$\mathcal{T}_{\Y}^\perp$ are all nonempty. Then
$(\mathcal{T}_\X \otimes \mathcal{T}_\Y)^{\perp\perp\circ}
\subseteq \mathcal{T}_{\X}^{\circ}\otimes \mathcal{T}_{\Y}^{\circ}$.
\end{lemma}

\begin{proof}
Let $c \in (\mathcal{T}_\X \otimes \mathcal{T}_\Y)^{\perp\perp\circ}$.
We prove that $c^1 \in \mathcal{T}_{\X}^{\circ}$
(and $c^2 \in \mathcal{T}_{\Y}^{\circ}$). Since 
$c \subseteq c^1 \otimes c^2$ and $c^1 \otimes c^2$ is strict in 
$(\mathcal{T}_\X \otimes \mathcal{T}_\Y)^{\perp\perp}$ by Lemma
\ref{l-11}, we will be able to conclude that $c = c^1 \otimes c^2
\in \mathcal{T}_{\X}^{\circ} \otimes \mathcal{T}_{\Y}^\circ$ by strictness 
of $c$.

Totality of $c^1$ is due to Lemma \ref{l-12}.
To show strictness, let $x \in c^1$, so that $(x,y) \in c$ for some $y$.
Since $c$ is strict, there is $\mathfrak{d} \in 
(\mathcal{T}_\X\otimes\mathcal{T}_\Y)^\perp$ such that $(x,y ) \in \mathfrak{d}$.
We then have $x \in \mathfrak{d}^1$ and 
$\mathfrak{d}^1 \in \mathcal{T}_\X^\perp$ by Lemma \ref{l-13}.
\QED
\end{proof}

We have established the internal completeness of $\otimes$.
Let us next proceed to connective $!$.

\begin{lemma}\label{l-2}
Given $\mathfrak{c}_1, \dots, \mathfrak{c}_n \in \mathcal{T}_{\X}^\perp$,
let
$$
\wedge_i \mathfrak{c}_i := \{ \{x_1, \dots, x_n\} \in \X: 
x_i \in \mathfrak{c}_i\ (1\leq i \leq n)\}.
$$
Then $\wedge_i \mathfrak{c}_i \in (!\mathcal{T}_{\X})^\perp$.
In particular, $\wedge\!
\mathfrak{c} := \{\{x\} \in \X: x \in \mathfrak{c}\} \in 
(!\mathcal{T}_{\X})^\perp$.
\end{lemma}

\begin{proof}
Each $\wedge_i \mathfrak{c}_i$ consists of finite cliques of $\X$,
namely tokens of $!\X$, which are pairwise incoherent in $!\X$. 
Indeed, given distinct $\{x_1, \dots, x_n\}, \{y_1, \dots, y_n\}$,
there is $i$ such that $x_i \neq y_i$ and $x_i, y_i \in \mathfrak{c}_i$.
We have $\neg(x_i \coh y_i)$, so that 
$\neg (\{x_1, \dots, x_n\} \coh \{y_1, \dots, y_n\})$.
Hence
$\wedge_i \mathfrak{c}_i \in (!\X)^\perp$.

To see totality, suppose that $a \in \mathcal{T}_\X$ so that $!a \in !\mathcal{T}_\X$.
We then have $x_i \in a \cap \mathfrak{c}_i$ for each $1\leq i \leq n$, hence
the clique $\{x_1, \dots, x_n\} \subseteq a$ belongs to $\wedge_i \mathfrak{c}_i$.
This proves $!a \perp \wedge_i \mathfrak{c}_i$.\QED
\end{proof}

%

\begin{lemma}\label{l-4}
$\,!\, \mathcal{T}_{\X}^{\circ} \subseteq 
(\,!\,\mathcal{T}_{\X})^{\perp\perp\circ}$. 
\end{lemma}

\begin{proof}
It is easy to see that
every $! a \in 
\,!\, \mathcal{T}_{\X}^{\circ}$ belongs to 
$(\,!\,\mathcal{T}_{\X})^{\perp\perp}$. For strictness, let 
$a_0 = \{x_1, \dots, x_n\} \in !a$. Since $a \in 
\mathcal{T}_{\X}^{\circ}$, there are
$\mathfrak{c}_i \in \mathcal{T}_{\X}^\perp$
such that 
$x_i \in \mathfrak{c}_i$ for each $1\leq i \leq n$.
Applying the previous lemma, we obtain 
$\wedge_i \mathfrak{c}_i \in
(\,!\,\mathcal{T}_{\X})^\perp$, which contains $a_0$.\QED
\end{proof}

\begin{lemma}\label{l-5}
Given $\alpha \in (!\mathcal{T}_\X)^{\perp\perp}$, let
$\alpha^1 
:= \{x : \{x\} \in \alpha \cap (\wedge\!\mathfrak{c}) \mbox{ for some }
\mathfrak{c} \in \mathcal{T}_\X^\perp\}$.
Then $\alpha^1 \in \mathcal{T}_\X^\circ$.
\end{lemma}

\begin{proof}
Given $\mathfrak{c} \in \mathcal{T}_\X^\perp$, we have 
$\wedge\!\mathfrak{c} \in (!\mathcal{T}_\X)^\perp$ by Lemma \ref{l-2}.
Hence $\{x\} \in \alpha \cap (\wedge\!\mathfrak{c})$ for some $x$,
so $x \in \alpha^1 \cap \mathfrak{c}$. Strictness of $\alpha^1$ is obvious.
\QED
\end{proof}

\begin{lemma}\label{l-6}
$(\,!\,\mathcal{T}_{\X})^{\perp\perp\circ} \subseteq \,!\, (\mathcal{T}_{\X}^{\circ})$. 
\end{lemma}

\begin{proof}
Let $\alpha \in (!\mathcal{T}_\X)^{\perp\perp\circ}$. 
We prove that $!\alpha^1 \subseteq \alpha$. Since 
$!\alpha^1\in (!\mathcal{T}_\X)^{\perp\perp\circ}$ by combining Lemmas \ref{l-4} and \ref{l-5},
we will be able to conclude $\alpha = !\alpha^1 \in 
!\mathcal{T}_\X^\circ$ by strictness of $\alpha$.

Let $a_0 \in !\alpha^1$. As in the proof of 
Lemma \ref{l-4}, we obtain 
$\mathfrak{c}_i \in \mathcal{T}_{\X}^{\perp}$ ($1\leq i \leq n$)
such that 
$a_0 \in \wedge_i \mathfrak{c}_i \in
(\,!\,\mathcal{T}_{\X})^\perp$.
Meanwhile, we have some $a_1 \in \alpha\cap \wedge_i \mathfrak{c}_i$ 
since $\alpha \in (!\mathcal{T}_\X)^{\perp\perp\circ}$. 
If $a_0 \neq a_1$, there would be  $x, y \in \mathfrak{c}_i$ 
such that $x \in a_0$, $y \in a_1$ and $x\neq y$. We would have 
$\neg (x \coh y)$ since they belong to an anti-clique $\mathfrak{c}_i$, 
while $x \coh y$ since they belong to a clique $\bigcup \alpha$, that is 
a contradiction. 
We therefore conclude that $a_0 = a_1 \in \alpha$.
\QED
\end{proof}

This completes the proof of Proposition \ref{prop-total-cliques}.

\subsection{The proof of Theorem \ref{thm-ubd-fine}}\label{A-thm-ubd-fine}
Let us begin with an important lemma: 
\begin{lemma}\label{lemma-upper-extension}
For every $a\in\X_\Fin$ with 
$\BO{a}^\circ \neq\emptyset$, 
there exist finitely many uni-covers 
$\mathfrak{c}_1 ,\ldots ,\mathfrak{c}_m \in \sigma^{\mathrm{b}}_\X$ 
such that $a\in \mathfrak{c}_1 \wedge \cdots \wedge \mathfrak{c}_m$, 
hence $a$ is contained in 
the unbounded cover $\mathfrak{c}_1 \wedge \cdots \wedge \mathfrak{c}_m \in \beta_{\X}^{\mathrm{ub}}$. 
\end{lemma}

\begin{proof}
Let $a=\{x_1 ,\ldots ,x_m\}$. 
By assumption, there exists $b\in\mathcal{T}_{\X}^{\circ}$ with $b\supseteq a$. 
By strictness of $b$, each $x_i$ is contained in some $\mathfrak{c}_i \in \mathcal{T}_{\X}^\perp$, which we may assume is strict. Hence we have
$a\in \mathfrak{c}_{1} \wedge \cdots \wedge  \mathfrak{c}_{m}$. \QED
\end{proof}

The next lemma is used to cut down and divide the set $\mathcal{T}_\X^\circ$. 
\begin{lemma} \label{lem_finite_divide}
For any $n\in \Nat$ and $a_0 , \ldots ,  a_n \in \X_\Fin$, 
there exists $\mathfrak{B}\in (!\X)^\perp$ 
such that 
$\sum_{b\in\mathfrak{B}} \BO{b}^\circ = \BO{a_n}^\circ\, \setminus\, \bigcup_{i=0}^{n-1}\BO{a_i}^\circ$. 
\end{lemma}
\begin{proof}
Without loss of generality, one can assume that $\BO{a_i}^\circ \neq\emptyset$ for all $i=1,\ldots ,n$. 
By Lemma \ref{lemma-upper-extension}, 
each $a_i$ belongs to some unbounded-cover $\mathfrak{A}_i \in\beta^{\mathrm{ub}}$. 
Define 
\begin{eqnarray*}
\mathfrak{B}\quad &:=&\quad \big(\mathfrak{A}_1\backslash \{a_0\}\big) 
\wedge \cdots\wedge \big(\mathfrak{A}_{n-1} \backslash \{a_{n-1}\}\big) \wedge \{a_n\} \\
&=&\quad \big\{ b_0 \cup \dots \cup b_{n-1} \cup a_n \in \X \ \left|\ b_i 
\in (\mathfrak{A}_i\backslash
\{a_i\})\ (0 \leq i \leq n-1)\big\}\right.^\circ.
\end{eqnarray*}
We then have $c \in \BO{b}^\circ$ for some 
$b= b_0 \cup \dots \cup b_{n-1} \cup a_n \in \mathfrak{B}$ 
iff $c \in
\BO{b_0}^\circ \cap \dots \cap \BO{b_{n-1}}^\circ \cap \BO{a_n}^\circ$ 
iff $c \not\in \BO{a_0}^\circ, \dots, 
c \not\in \BO{a_{n-1}}^\circ$ (by disjointness) and $c \in \BO{a_n}^\circ$.\QED
\end{proof}

Now we go on the proof 
of Theorem \ref{thm-ubd-fine}:
$(\mathcal{T}_{\X}^\circ, \mu^{\mathrm{ub}}_\X)$ is a fine uniform space. 
\begin{proof}
Let $\mathcal{U}$ be an open cover of $\mathcal{T}_\X^\circ$. 
Since an open set in $\mathcal{T}_\X^\circ$ is 
a countable union of upper sets $\BO{a}^\circ$, 
we can assume that $\mathcal{U}$ is of the form $\{\BO{a_n}^\circ : n\in \Nat\}$. 
Our goal is to show 
that there exists $\mathfrak{B}\in \mathcal{T}_{!\X}^{\perp}$ which refines $\mathcal{U}$: 
for any $b\in\mathfrak{B}$ there exists $n\in \Nat$ with $\BO{b}\subseteq \BO{a_n}$. 
It will follow that
$\mathfrak{B}^\circ$ belongs to $\mu^{\mathrm{ub}}$ and so does 
$\mathcal{U}$ by axiom 
{\rm ({\bf U2})}. 

Let us denote 
$\mathcal{D}_n := \BO{a_n}^\circ \,\setminus\, \bigcup_{i=0}^{n-1}
\BO{a_i}^\circ$ so that 
$\mathcal{T}_\X^\circ =\sum_n \mathcal{D}_n$.
For each $n\in\Nat$, 
apply Lemma \ref{lem_finite_divide} to 
$a_0, \ldots , a_{n}$ 
to find $\mathfrak{B}_n \in (!\,\X)^\perp$ such that 
$\sum_{b\in\mathfrak{B}_n } \BO{b}^\circ = \mathcal{D}_n$.

Now it is easy to see that 
$\mathcal{B} := \bigcup_n \mathcal{B}_n$ also belongs to 
$(!\,\X)^\perp$ 
and moreover $\mathcal{B} \in \mathcal{T}_{!\X}^{\perp}$, 
since $\{\BO{b}^\circ:b\in\mathfrak{B}\}$ covers $\mathcal{T}_\X^\circ = \sum_n \mathcal{D}_n$. 
Finally, any $b \in \mathcal{B}_n \subseteq \mathcal{B}$ 
satisfies $\BO{b}\subseteq \mathcal{D}_n \subseteq \BO{a_n}$.  \QED
\end{proof}

\subsection{Existence of a pseudo-Map of Adjunctions}\label{A-pm}
The pair of functors $\BO{I,J}$ in the diagram (\ref{fig-pseudo}) 
is a \emph{psudo-map of adjunctions}. 
That is, 
$\BO{I,J}$ preserves the counit-unit pair (up to isomorphisms): 
$J\eta_{\mathrm{coh}}=(\eta_{\mathrm{unif}})_{J}: J\Rightarrow JLK\simeq GFJ$ 
and $I\epsilon_{\mathrm{coh}} =(\epsilon_{\mathrm{unif}})_{I}:IKL\simeq FGI\Rightarrow I$, 
where $\BO{\epsilon_{\mathrm{coh}},\eta_{\mathrm{coh}}}$ 
and $\BO{\epsilon_{\mathrm{unif}},\eta_{\mathrm{unif}}}$ are the counit-unit pairs 
of the adjunctions $F\dashv G$ and $K\dashv L$, respectively. 

\begin{proof}
Let $\X \in \mathbf{Stab}_{\STot}$ be a coherence space with totality. 
Since $(\eta_{\mathrm{coh}})_\X (a):=\,!\,a$ for every $a\in\X$, 
$(I\eta_{\mathrm{coh}})_\X : (\mathcal{T}_{\X}^\circ ,\tau_{\mathrm{Sco}}) \to (\mathcal{T}_{\,!\,\X}^\circ, \tau_{\mathrm{Sco}})$ gives the bijection $a\in \mathcal{T}_{\X}^\circ \mapsto \,!\,a \in\mathcal{T}_{\,!\,\X}^\circ$ given in Proposition \ref{prop-total-cliques}.
On the other hand, $(\eta_{\mathrm{unif}})_{I\X}: (\mathcal{T}_{\X}^\circ ,\tau_{\mathrm{Sco}}) \to 
(\mathcal{T}_{\X}^\circ ,\tau_{\mathrm{Sco}})$ is the identity, 
therefore, $I\eta_{\mathrm{coh}}= (\eta_{\mathrm{unif}})_I$ up to isomorphisms. 

Let $\Y\in\Lin_{\STot}$ be a coherence space with totality.  
Recall that $(\epsilon_{\mathrm{coh}})_\Y :\,!\, \Y \Linear \Y$ is the \emph{dereliction} 
so that $(\epsilon_{\mathrm{coh}})_\Y (\,!\, a) := a$ for every $a\in \Y$. 
Hence $(J\epsilon_{\mathrm{coh}})_\Y : (\mathcal{T}_{\,!\,\Y}^\circ ,\mu_{\,!\,\Y}^{\mathrm{b}}) 
\to (\mathcal{T}_{\Y}^\circ, \mu_{\Y}^{\mathrm{b}})$ 
also gives the uniform homeomorphism in Proposition \ref{prop-unif-hom}. 
Similarly, $(\epsilon_{\mathrm{unif}})_{J\Y}: (\mathcal{T}_\Y ^\circ ,\mu_{\mathrm{fine}})\to (\mathcal{T}_\Y ^\circ,\mu_{\Y}^{\mathrm{b}})$ is the identity, 
therefore, we obtain $J\epsilon_{\mathrm{coh}} = (\epsilon_{\mathrm{unif}})_J $ up to isomorphisms. \QED
\end{proof}

\subsection{Uniform Structure on a Linear Function Space}
We shall exhibit an explicit structure of uniformity on a function space induced by co-totality given in \S 3. 
\begin{proposition}\label{prop-lin-func}
Every total linear map $F:(\X\llto \Y)\Linear \Z$ is 
uniformly continuous at single points: 
for every $\mathfrak{c}\in\sigma_{\Z}^{\mathrm{b}}$, 
there exist $a\in\mathcal{T}_\X^\circ$ and a total linear map $G:\Y\Linear \Z$ 
such that $|F(\kappa)-G(\widehat{\kappa}(a_0))|<\mathfrak{c}$ for every $\kappa \in \mathcal{T}_{\X\llto\Y}^\circ$. 
\end{proposition}

\begin{proof}
By Theorem \ref{theo-strongly-unif}, 
there exists $\mathfrak{f}\in \sigma^{\mathrm{b}}_{\X\llto\Y}$ such that 
$|\kappa -\kappa'|<\mathfrak{f} \ \Rightarrow \ |F(\kappa) -F(\kappa')|<\mathfrak{c}$ 
for all $\kappa,\kappa'\in \mathcal{T}_{\X\llto\Y}^\circ$. 
Notice that $(\X\llto\Y)^\perp = \X\otimes \Y^\perp$. 
Hence $\sigma^{\mathrm{b}}_{\X\llto\Y} = \mathcal{T}_\X^\circ \otimes (\mathcal{T}_\Y^\perp)^\circ$ by Proposition \ref{prop-total-cliques}, 
so $\mathfrak{f}= a_0 \otimes \mathfrak{b}$ for some $a_0\in \mathcal{T}_\X^\circ$ and $\mathfrak{b}\in\sigma^{\mathrm{b}}_\Y$. 
We now have $|\kappa -\kappa'|<a_0 \otimes \mathfrak{b} \Rightarrow |F(\kappa) -F(\kappa')|<\mathfrak{c}$.

Let $\mathfrak{c}\in\mathcal{T}_{\X}^\circ$ be an arbitrary uni-cover, 
and define $\theta \in \Y\Linear (\X\llto\Y)$ as
$\Tr(\theta):= \{ (y,(x,y)): \mbox{$x\in\mathfrak{c}$, $y\in Y$}\}$. 
It is easy to see that $\theta$ is strictly total and 
satisfies that $\widehat{\theta}(b)(a):=b$ for all $a\in \mathcal{T}_\X$ and $b\in\mathcal{T}_\Y$. 

Then $G:= F\circ \theta: \Y\Linear \Z$ is a total linear map satisfying our requirement: 
By letting $b_0:=\widehat{\kappa}(a_0)$, 
$$
\widehat{\kappa}(a_0)=\widehat{\theta}(b_0)(a_0)= b_0 \ \Longrightarrow \ 
|\kappa - \theta(b_0)|<a_0 \otimes \mathfrak{b} \ \Longrightarrow \ 
|F(\kappa) -F(\theta(b_0))|<\mathfrak{c}\ . 
$$
\QED
\end{proof}

What is interesting here is that the uniform structures on the constructed spaces 
are determined by purely logical rules: for instance $(\X\llto \Y)^\perp = \X\otimes \Y^\perp$.

\subsection{The proof of Theorem \ref{thm-unif-realize}}\label{A-uni-realize}
To prove the theorem, 
we first observe that a standard representation $\delta_\XX$ satisfies a kind of universality. 

A coherent representation $\Rep{\X}{\gamma}{\YY}$ is said to be \emph{linearish} if 
(i) $x\coh y$ implies $\gamma[\BO{x}]\cap\gamma[\BO{y}]\neq\emptyset$ for every $x,y\in X$; 
and (ii) for every uniform cover $\mathcal{U}\in\mu_\YY$ there exists $\mathfrak{c}\in \sigma^{\mathrm{b}}_\X$ 
such that $|a-b|<\mathfrak{c} \Rightarrow |\gamma(a)-\gamma(b)|<\mathcal{U}$ for all $a,b\in\Dom(\gamma)$, 
which is an analogue of strong uniform continuity in Theorem \ref{theo-strongly-unif}. 

A standard representation $\Rep{\B_\XX}{\delta_\XX}{\XX}$ is indeed linearish. 
Let $\beta=\{\mathcal{U}_n\}$ be a countable basis of a uniform space $\XX = (X,\mu_X)$ 
and $\delta_\XX$ is a standard representation induced from $\beta$. 
(i) By definition, $(n,U)\scoh (m,V)$ implies $U\cap V\neq\emptyset$, 
hence $\delta_\XX [(n,U)] = U$, $\delta_\XX [(m,V)]=V$ and 
$\delta_\XX [(n,U)]\cap \delta_\XX [(m,V)] \neq\emptyset$. 
(ii) For every uniform cover $\mathcal{U}\in\mu_X$, 
there exists $n\in\Nat$ 
such that $\mathcal{U}_n$ refines $\mathcal{U}$. 
Since $\mathfrak{c}:=\{(n,U):U\in\mathcal{U}_n\}$ is a uni-partition of $\mathcal{T}_{\B_\XX}$, 
we have $|a-b|<\mathfrak{c} \Rightarrow |\delta_\XX (a)-\delta_\XX(b)| <\mathcal{U}_n <\mathcal{U}$. 

The following lemma indicates that $\delta_\XX$ is a representative example of linearish representations. 
\begin{lemma}\label{prop-admissible}
For any subspace $\XX_0\subseteq \XX$ and 
any linearish representation $\Rep{\X}{\gamma}{\XX_0}$, 
there exists a linear map $F:\X\Linear \B_\XX$ with $\delta_\XX\circ F = \gamma$. 
\end{lemma}
In particular, it immediately follows that 
standard representations of $\XX$ are all isomorphic by letting $\XX_0 :=\XX$, hence 
they do not depend on the choice of uniform basis $\beta$. 
\begin{proof}
Let $\{\mathcal{U}_n:n\in\Nat\}$ be a countable basis on $\XX$. 
Since $\gamma$ is linearish, we can take a sequence of uni-covers $\{\mathfrak{c}_n :n\in\Nat\}\subseteq \sigma^{\mathrm{b}}_\X$ 
such that $|a-b|<\mathfrak{c}_n \Rightarrow |\gamma(a)-\gamma(b)|<\mathcal{U}_n$ 
for each $n\in\Nat$. 
Let $\psi:\subseteq X \times \Nat\to B$ be a partial function so that 
$\psi(x,n):=U$ is defined for each $n\in\Nat$ and $x\in\mathfrak{a}_{n}$, 
and then $U\in\mathcal{U}_n$ and $\gamma[\BO{x}]\subseteq  U$. 
We define a linear map $F:X\Linear\B_\XX$ as 
$F(a):= \{\psi(x,n): \mbox{$n\in\Nat$ and $x\in a$}\}$. 
Let us now verify that 
$F$ is the desired map in 4 steps. 

(i) $F(a)\in \B_\XX$ for every $a\in\X$. 
Let $(n,U), (m,V)\in F(a)$ with $(n,U) \neq (m,V)$. 
This means that 
there exist 
$x\in \mathfrak{a}_n$, and $w\in \mathfrak{a}_m$ 
such that $\psi(x,n)=U$, $\psi(w,m)=V$ and 
$x,w\in a$ (so $x\coh w$). 
If $n=m$ then $x=w$ since $x\coh w$ but $x,w\in\mathfrak{a}_n$. 
Hence $U=\psi(x,n)=\psi(z,m)=V$, contradicting 
the assumption, so $n\neq m$. 

We also have $U\cap V\neq\emptyset$, 
since $x\coh w$ implies $\gamma[\BO{x}] \cap \gamma[\BO{w}] \neq\emptyset$, 
$\gamma[\BO{x}] \subseteq U$ and 
$\gamma[\BO{w}] \subseteq V$. 
Therefore, $(n,U)\scoh (m,V)$. 

(ii) $F$ is a linear map. 
It is sufficient to verify the condition that 
$F(a) \ni y$ implies the unique existence of $x\in X$
such that $F(\{x\}) \ni y$. Let $y =(n,U) \in B$. 
By definition, it is immediate that there is a unique $x\in a\cap \mathfrak{a}_n$ 
such that $\psi(x,n) = U$. 

(iii) $F(a) \in (\B_\XX)_\Max$ for every 
$a\in \Dom(\gamma)$. Let $x_n \in a \cap \mathfrak{a}_n$ for each $n\in\Nat$. 
It suffices to show that for every $n\in \Nat$ there is 
$(n,U) \in F(a)$ for some $U\in\mathcal{U}_n$. 
By definition, $\psi(x_n,n):= U$ is defined so that $(n,U)\in F(a)$. 

(iv) $\delta_{\XX} \circ F(a) = \gamma(a)$ for every $a\in \Dom(\gamma)$.
Suppose that $(n,U_n) \in F(a)$, namely 
there is $x \in a$ 
such that $\psi(x, n) = U_n$. By definition, we have
$ \gamma(a) \in \gamma[\BO{x}] \subseteq U_n$.
Since it holds for every token of $F(a)$, we have
$\gamma(a) \in \bigcap_{(n,U_n)\in F(a)} U_{n}$.  
Therefore, $F(a)\in\Dom(\delta_\XX)$ and $\delta_{\XX} \circ F(a) = \gamma(a)$. \QED
\end{proof}

We are now ready to prove Theorem \ref{thm-unif-realize}. 
Let $\YY_0:= f[\XX]$ and $\gamma:= f\circ \delta_\XX$. 
Then $\gamma$ is linearish: 
(i) If $x\coh y$ then $\delta_\XX [\BO{x}] \cap \delta_\XX [\BO{y}]\neq\emptyset$ 
hence $\gamma[\BO{x}]\cap\gamma[\BO{y}]\neq\emptyset$. 
(ii) For every (subspace) uniform cover $\mathcal{V}\in\mu_{\YY_0}$, 
there exists $\mathcal{U}\in\mu_\XX$ such that 
$|p-q|<\mathcal{U}\ \Rightarrow \ |f(p)-f(q)|<\mathcal{V}$, 
and we also have $\mathfrak{c}\in\sigma^{\mathrm{b}}_\X$ such that 
$|a-b|<\mathfrak{c}\ \Rightarrow \ |\delta_\XX (a)-\delta_\XX (b)|<\mathcal{U}$, 
since $\delta_\XX$ is linearish. 

Applying Lemma \ref{prop-admissible} to $\B_\XX\stackrel{\delta_\XX}{\longrightarrow}\XX\stackrel{f}{\longrightarrow}{\YY_0}$, 
we obtain a (total) linear map 
$F:\B_\XX \Linear\B_\YY$ such that $\delta_\YY \circ F = f\circ \delta_\XX$. 
This concludes Theorem \ref{thm-unif-realize}. 

\subsection{The Lemma for Theorem \ref{thm-real-unif}}\label{A-lem-quot}
\begin{lemma}\label{lem-quot}
If $\XX$ is chain-connected, the standard representation $\delta_{\XX}$ is a uniform quotient. 
That is, $\{\delta_\XX [\mathfrak{c}]:\mathfrak{c}\in \sigma_{\X}^{\mathrm{b}}\}$ 
is a uniform basis of $\XX$, 
where $\delta_\XX[\mathfrak{c}]$ is a cover of $\XX$ defined by $\{\delta_\XX[\BO{x}]:x\in\mathfrak{c}\}$. 
\end{lemma}

\begin{proof}
We need to check that
the surjection $\Rep{\B_\XX}{\delta_\XX}{\XX}$ induces the uniformity on $\XX$: 
namely, $\{\delta_{\XX}[\mathfrak{c}]:\mathfrak{c}\in \sigma^{\mathrm{b}}_{\B_{\XX}}\}$ 
forms a uniform basis of $\XX$, 
where $\delta_\XX [\mathfrak{c}]$ is a cover of $\XX$ defined by
$\delta_\XX[\mathfrak{c}]:= \{\delta_\XX[\BO{x}]:x\in\mathfrak{c}\}$. 

Let $\{\mathcal{U}_n :n\in\Nat\}$ be a countable basis of $\XX$. 
All we have to show is that 
each uni-cover $\mathfrak{c}\in\sigma^{\mathrm{b}}_{\B_\XX}$ is of the form 
$\mathfrak{c}=\{(n,U):U\in\mathcal{U}_n\}$ for some $n\in\Nat$. 
Then $\delta_\XX [\mathfrak{c}]=\mathcal{U}_n$, hence they generate the uniformity on $\XX$. 

Let $\mathfrak{c}\in\sigma^{\mathrm{b}}_{\B_\XX}$ be a uni-cover of $\Dom(\delta_\XX)$. 
Since $\mathfrak{c}\neq\emptyset$, one can fix a token $(n,U_n)\in\mathfrak{c}$. 

Indeed $\mathfrak{c}= \{(n,U):U\in\mathcal{U}_n\}$. 
By chain-connectedness of $\XX$, 
we have $U_{n}'\neq U_n\in\mathcal{U}_n$ such that $U_n \cap U_{n}'\neq\emptyset$. 
Given arbitrary $p\in U_n \cap U_{n}'$, we can take 
$\{U_m : m\neq n\}$ such that $p \in U_m \in \mathcal{U}_m$ for each $m\neq n\in\Nat$. 
Let $a:= \{(m,U_m):m\neq n\}\cup \{(n,U_n)\}$ and 
$a':=\{(m,U_m):m\neq n\}\cup \{(n,U'_n)\}$. 
Both $p\in \bigcap_{m\neq n} U_m \cap U_n$ and $p\in \bigcap_{m\neq n} U_m \cap U'_n$ hold, 
hence $a,a'\in\Dom(\delta_\XX)$. 
Since $\mathfrak{c}$ is a uni-cover of $\Dom(\delta_\XX)$, 
both $a\perp \mathfrak{c}$ and $a'\perp \mathfrak{c}$ hold, 
therefore, $(n,U_n')\in\mathfrak{c}$ 
(otherwise, if $(m,U_m) \in a'\cap \mathfrak{c}$, then 
$a\cap \mathfrak{c}$ contains both $(n,U_n)$ and $(m,U_m)$, 
which contradicts that $a\cap \mathfrak{c}$ must be a singleton. 

Repeating this argument, we obtain $(n,U)\in \mathfrak{c}$ for all $U\in\mathcal{U}_n$. \QED
\end{proof}

\subsection{Some Constructions of Coherent Representations}\label{A-classical-rep}
Typical constructions of coherent representations are naturally given as follows. 
Given $\Rep{\X}{\rho_\X}{S}$ and $\Rep{\Y}{\rho_\Y}{T}$, define: 
\begin{itemize}
\item $\Rep{\X\otimes \Y}{[\rho_\X \otimes\rho_\Y]}{S\times T}$
is defined as 
$\Dom([\rho_{\X}\otimes\rho_{\Y}]) := \Dom(\rho_\X)\otimes \Dom(\rho_\Y)$ and 
$[\rho_{\X}\otimes\rho_{\Y}] (a\otimes b): = (\rho_\X (a) ,\rho_\Y (b))$, 
where $\Dom(\ )$ means the domains of representations (as partial maps). 
\item $\Rep{\X\llto \Y}{[\rho_{\X}\llto \rho_\Y]}{\LR(\rho_\X ,\rho_\Y )}$ is 
defined as follows. Define $[\rho_{\X}\llto \rho_\Y]: \subseteq \X\llto \Y 
\longrightarrow T^S$ by
$$[\rho_{\X}\llto \rho_\Y](\kappa):= f \Iff
f:S\longrightarrow T \mbox{ is realized by 
$\widehat{\kappa} : \X\Linear \Y$.}
$$
$\LR(\rho_\X ,\rho_\Y ) \subseteq \YY^\XX$ is the range of $[\rho_{\X}\llto \rho_\Y]$, which consists of 
linearly realizable functions.
\item $\Rep{\,!\,\X}{[\,!\,\rho_\X]}{S}$ is 
defined as $\Dom([\,!\,\rho_{\X}]) := \,!\,\Dom(\rho_\X)$
and 
$[\,!\,\rho_{\X}](\,!\,a):= \rho_\X (a)$ 
for every $a\in\Dom(\rho_\X)$. 
\end{itemize}

Unfortunately, 
the total extension lemma (Lemma \ref{lem-tot-ext}) is no longer available for 
these constructions. 
For instance, we do not have 
$\Dom([\rho_{\X}\otimes \rho_{\Y}])^{\perp\perp} = \mathcal{T}_{\X\otimes\Y}$ in general, 
where $\mathcal{T}_{\X\otimes\Y}$ 
is the tensor of the totalities $\mathcal{T}_\X$ and $\mathcal{T}_\Y$ 
which are the double negations of $\Dom(\rho_\X)$ and $\Dom(\rho_\Y)$ respectively. 

To avoid this, we consider the following condition. 
A coherent representation $\Rep{\X}{\rho_\X}{S}$ 
is said to be \emph{classical} if 
$\Dom(\rho_\X)= \Dom(\rho_\X)^{\perp\perp\circ}$
(i.e. the domain is a totality on $\X$ and consists of strict total cliques). 
As noted in Example \ref{ex-standard-totality}, 
a complete space $\XX$ has a classical standard representation $\delta_\XX$. 

Then it is easy to see that 
if $\rho_\X$ and $\rho_\Y$ are classical, 
so are $[\rho_\X \otimes \rho_\Y]$ and $[\,!\,\rho_\X]$, 
due to the internal completeness (Proposition \ref{prop-total-cliques}). 
Although $[\rho_\X \llto \rho_\Y]$ is not classical, 
one can naturally restrict it as follows. 
Since $[\rho_\X \llto \rho_\Y](\kappa)= [\rho_\X \llto \rho_\Y](\kappa^\circ)$
for every $\kappa\in\mathcal{T}_{\X\llto\Y}$, 
the strict restriction $[\rho_\X \llto \rho_\Y]^\circ : 
\mathcal{T}_{\X\llto\Y}^\circ \to \mathcal{LR}(\rho_\X,\rho_\Y)$ 
is well-defined. 

These representations are indeed compatible with the uniformities induced by totalities:
$\Dom([\rho_\X \otimes \rho_\Y])^{\perp\perp}= \mathcal{T}_{\X\otimes\Y}$, 
$\Dom([\rho_\X \llto \rho_\Y]^\circ)^{\perp\perp}=\mathcal{T}_{\X\llto\Y}$ and 
$\Dom([\,!\,\rho_\X])^{\perp\perp}= \mathcal{T}_{\,!\,\X}$, 
where $\mathcal{T}_\X := \Dom(\rho_\X)$ and $\mathcal{T}_\Y:=\Dom(\rho_\Y)$. 
\end{document}